\documentclass[prl,twocolumn,longbibliography,nofootinbib, preprintnumbers,notitlepage,fleqn,showpacs]{revtex4-1}

\usepackage{color}

\usepackage[bookmarks = true, citecolor = blue, colorlinks = true, linkcolor = magenta, urlcolor = blue]{hyperref}

\usepackage[T1]{fontenc}			
\usepackage[sc,osf]{mathpazo}   	

\usepackage{amsmath}  			
\usepackage{amsfonts}  			
\usepackage{graphicx}   			
\usepackage{mathrsfs, amsthm, amssymb}
\usepackage{bbm, bm}

\usepackage{nicefrac}    			

\usepackage[braket, qm]{qcircuit}	

\newtheorem{theorem}{Theorem}
\newtheorem{definition}{Definition}

\newtheorem{lemma}{Lemma}

\begin{document}


\title[]{Is single-particle interference spooky?}

\author{Pawel \surname{Blasiak}}
\email{Pawel.Blasiak@ifj.edu.pl}
\affiliation{Institute of Nuclear Physics Polish Academy of Sciences, PL-31342 Krak\'ow, Poland}


\begin{abstract}
It is said about quantum interference that \emph{"In reality, it contains the only mystery"}. Indeed, together with non-locality it is often considered as the characteristic feature of quantum theory which can not be explained in any classical way. In this work we are concerned with a restricted setting of a single particle propagating in multi-path interferometric circuits, that is physical realisation of a qudit. It is shown that this framework, including collapse of the wave function, can be simulated with classical resources without violating the locality principle. We present a \emph{local ontological model} whose predictions are indistinguishable from the quantum case. 'Non-locality' in the model appears merely as an epistemic effect arising on the level of description by agents whose knowledge is incomplete. This result suggests that the real quantum mystery should be sought in the multi-particle behaviour, since single-particle interferometric phenomena are explicable in a classical manner.
\end{abstract}


\maketitle

In the Feynman Lectures on Physics quantum interference is described as \emph{"a phenomenon which is impossible, \emph{absolutely} impossible, to explain in any classical way, and which has in it the heart of quantum mechanics"}~\cite{FeLeSa65}. Broadly speaking, the phenomenon concerns behaviour of a particle in the interferometric circuits and the problem consists in reconciling wave and particle character of the phenomenon. Another difficulty is a common-sense explanation of the collapse of the wave function upon measurement. In some mysterious way behaviour of the quantum particle depends on the knowledge of what is happening in the distant parts of the experimental setup. Notably, non-locality of the collapse manifests itself already in the single-particle scenarios, as first pointed out by A. Einstein during the Fifth Solvay Conference~\cite{BaVa09} who metaphorically called such an influence \emph{"spooky action-at-a-distance"}~\cite{EiPoRo35,Ei71}. A fully fledged argument against local realism in quantum theory is due to profound insight of J. S. Bell~\cite{Be66,Be93}. It requires two particles to show non-local correlations between measurements in distant arms of the interferometric setup. Remarkably, all further refinements of the argument exploit properties of entangled states in multi-particle scenarios, see e.g.~\cite{GrHoShZe90,Ha92a,Me93}. This leaves open the question of possible local explanation of quantum interferometric phenomena in the single-particle case, cf.~\cite{TaWaCo91,Ha94,GrHoZe95,DuVe07}.

Quantum mechanics of single-particle phenomena is a rich source of paradoxes and surprising effects which challenge our classical intuition about the world. Apart from quantum interference~\cite{ScSu98}, they include e.g.: interaction-free measurements~\cite{ElVa93,KwWeHeZeKa95,De02}, quantum Zeno effect~\cite{MiSu77,Pe80,KwWeHeZeKa95,De02}, Wheeler's delayed-choice experiment~\cite{Wh78,JaWuGrTrGrAsRo07}, violation of Leggett-Garg inequalities~\cite{LeGa85,EmLaNo14}, pre- and post-selection paradoxes~\cite{AhAlVa88,GeRoMaBlBeMaTw13} and contextuality~\cite{KoSp67,KlCaBiSh08}. These phenomena are often considered as strictly quantum mechanical effects and some of them, like contextuality or Leggett-Garg inequalities, are sometimes treated as signatures of the quantum regime. However, as suggestive it might look it is not at all clear to what extent these features are unique to the quantum realm. On the one hand, there are various models indicating analogies on the grounds of classical probabilistic theories, see e.g.~\cite{Ha99,DaPlPl02,Sp07,KlGuPoLaCa11,BaRuSp12,WaBa12,Bl13,FeCo14,Bl15,Bl15a,KaCaBaRu15}. On the other hand, none of these results fully reconstruct quantum predictions for general single-particle scenarios. All this makes the question about the distinctive quantum features an interesting problem. In particular, it is not clear whether non-locality in the single-particle framework is on a par with the multi-particle case, i.e. does not admit explanation via local hidden variable models~\cite{TaWaCo91,Ha94,GrHoZe95,DuVe07}. A decisive answer would require either a rigorous no-go proof, like the Bell's theorem is for two particles, or a counterexample encompassing all relevant aspects of quantum interferometric setups.

In this paper we are concerned with a single particle propagating in general multi-path interferometric circuits -- that is physical realisation of finite dimensional Hilbert space $\mathcal{H}\!=\!\mathbb{C}^N$ (qudit)~\cite{ReZeBeBe94} -- and explicitly construct local ontological model which faithfully imitates all quantum mechanical predictions. This suggests caution against statements to the effect of non-locality of the collapse of the wave function or \emph{absolute} impossibility of classical explanation of single-particle interferometric phenomena. The model shows that local explanation is conceivable and the real mystery lays in the multi-particle behaviour~\cite{Be66,Be93,GrHoShZe90,Ha92a,Me93}. Analysis of the model illustrates the role of epistemic constraints in description of the system by agents with limited resources, which lead to all kinds of weird quantum-like effects.


\section{Results}
We begin with a brief account of  quantum interferometric circuits. It is meant to introduce the notation and provide a basis for comparison of the model constructed in this paper with the standard quantum mechanical description.

\ \\\noindent\textbf{
Quantum interferometry in a nutshell}.
In the following we consider standard interferometric framework for a single particle propagating through a network of spatially separated paths. Evolution of the system is implemented by gates attached to the paths which represent non-trivial transformations (with empty paths corresponding to free evolution). See Fig.~\ref{Fig-Ontology} (on the right) for illustration. It is enough to consider only a few kinds of gates which form a basis for construction of complex interferometric circuits~\cite{ReZeBeBe94}. These gates include phase shifters $S_j$ and detectors $D_j$ which are attached to individual paths and beam splitters  $B_{st}$ on which two paths cross, with $j$ and $s,t$ indicating the respective paths. A special role of detectors is to provide an outcome \textsc{'Click'}\!/\textsc{'No Click'} which attests to the presence/absence of the particle in a given path. 


Quantum description of a single particle in the interferometric circuit which consist of $N$ paths associates position of the particle with the vectors of computational basis $\ket{1},\dots,\ket{N}$, where $\ket{j}$ represents the fact of particle being in $j$-th path. In general, state of the system is a superposition with complex coefficients $\psi_j$ defining a vector (ray) in $\mathcal{H}=\mathbb{C}^N$, i.e.
\begin{eqnarray}\label{Quant-State}
\ket{\psi}=\sum_{j=1}^N\ \psi_j \ket{j}=\left(\begin{smallmatrix}\psi_1\vspace{-0.1cm}\\\vdots\vspace{0.1cm}\\\psi_N\end{smallmatrix}\right)=\vec{\psi},
\end{eqnarray}
with normalisation $\Vert\!\ket{\psi}\!\Vert^2=\sum_j |\psi_j|^2=1$ and vectors differing be an overall phase being equivalent. Evolution implemented by gates corresponds to a sequence of unitary and projective transformations described as follows. \emph{Free evolution} in $j$-th path acts trivially and \emph{phase shifter} $S_j$ introduces phase $e^{i\omega}$ in the relevant path, i.e.
\begin{eqnarray}\label{Quant-Free+Pj}
\psi_j\ \xymatrix{\ar[r]^{\textit{free}} &}\ \psi_j&\ \ \ \text{and}\ \ \ &\psi_j\ \xymatrix{\ar[r]^{S_j} &}\ e^{i\omega}\,\psi_j\,.
\end{eqnarray}
\emph{Beam splitter} $B_{st}$ located at the crossing of paths $s$ and $t$ implements a unitary in the subspace spanned by kets $\ket{s}$ and $\ket{t}$ given by
\begin{eqnarray}\label{Quant-Bjk}
\begin{pmatrix}\psi_s\\\psi_t\end{pmatrix}\ \xymatrix{\ar[r]^{B_{st}} &}\ \begin{pmatrix}\psi_s'\\\psi_t'\end{pmatrix}=\begin{pmatrix}i\sqrt{R}&\sqrt{T}\\\sqrt{T}&i\sqrt{R}\end{pmatrix}\!\!\begin{pmatrix}\psi_s\\\psi_t\end{pmatrix}\,,
\end{eqnarray}
where $R$,$T$ are reflectivity and transitivity coefficients. 
Finally, according to the measurement postulate (von Neumann--L\"uders rule) \emph{detector} $D_j$ is described by the PVM $\{\mathbb{P}_j\,,\,\mathbb{1}-\mathbb{P}_j\}$ where $\mathbb{P}_j\equiv\ket{j}\!\bra{j}$, i.e. depending on the outcome it effects the projection
\begin{eqnarray}\label{Quant-Dj}
\ket{\psi}\ \xymatrix{\ar[r]^{D_j} &}\ \left\{
\begin{array}{cll}
\ket{j}&&\ \ \text{\textsc{'Click'}}\,,\vspace{0.1cm}\\
\frac{(\mathbb{1}-\mathbb{P}_j)\ket{\psi}}{\Vert(\mathbb{1}-\mathbb{P}_j)\ket{\psi}\Vert}&&\ \ \text{\textsc{'No Click'}}\,,
\end{array}
\right.
\end{eqnarray}
with probability that detector $D_j$ \textsc{'Clicks'} given by the Born rule $\textit{\textsf{Pr}}\,(D_j|\psi)=|\!\ip{j}{\psi}|^2=|\psi_j|^2$. Note that projection postulate Eq.\,(\ref{Quant-Dj}) affects the whole space $\mathcal{H}=\mathbb{C}^N$ in spite of the fact that detector $D_j$ is localised only in $j$-th path. Explanation of this behaviour leads to the notorious problem concerning ontological status of quantum states and the issue of non-locality of the collapse of the wave function.

These rules provide mathematical description of a single particle the interferometric circuit. It was shown in Ref.~\cite{ReZeBeBe94} that any unitary and projective measurement in $\mathcal{H}=\mathbb{C}^N$ can be experimentally realised in a circuit composed of $N$ paths as a sequence of interferometric gates defined above. Thus it provides a convenient physical framework for foundational explorations.


Our main goal in this paper is explicit construction of a classical analogue with the same structural components (comprised of paths and gates arranged into circuits) which mimics quantum behaviour of a particle in the interferometric circuits described above. The crux of the matter  is to provide a model with well-defined underlying ontology which does not violate the locality principle, and yet on the operational level its predictions are indistinguishable form the quantum case.
\ \\\noindent\textbf{Ontology of the model}.
Let us consider circuits composed of $N$ paths labelled with index $j=1,\dots,N$. Defining the model we assume that in the circuit propagates a \emph{single particle} which has well defined \emph{position} $q=1,\dots,N$. Additionally, we postulate that along each path propagates a \emph{local field} characterised by two degrees of freedom: (complex) \emph{amplitude} $u_j$ such that $|u_j|\leqslant1$ and (real) \emph{strength} $\tau_j$ such that $0\leqslant\tau_j\leqslant1$. This means that at each time the system of $N$ paths is fully specified by a point $(q,\vec{u},\vec{\tau})$ in the \emph{ontic state space}
\begin{eqnarray}\nonumber
\Lambda&=&\{q:q=1,\dots,N\}\times\\
&&\{\vec{u}\in\mathbb{C}^N:|u_j|\leqslant1\}\times\\\nonumber
&&\{\vec{\tau}\in\mathbb{R}^N:0\leqslant\tau_j\leqslant1\}\,,
\end{eqnarray}
where $u_j$ and $\tau_j$ describe the field in $j$-th path. See Fig.~\ref{Fig-Ontology}.

In the following we will be interested in stochastic evolution which requires probabilistic description and hence consider the set of all possible probability distributions over the ontic states\begin{eqnarray}
\!\!\!\mathcal{P}(\Lambda)=\Big\{\,\bm{p}:\Lambda\longrightarrow[0,1]\,:\,\int_{\Lambda}\bm{p}(\lambda)\,d\lambda=1\,\Big\},
\end{eqnarray}
which will be called \emph{epistemic state space}.
A general stochastic \emph{transformation (or gate)} is defined as a mapping $T:\Lambda\longrightarrow\mathcal{P}(\Lambda)$, where $T(\lambda)$ specifies distribution of final states given the system was in state $\lambda\in\Lambda$. In the model we will be concerned with a limited choice of transformations (gates) which are described below.

\begin{figure*}[t]
\begin{center}
\includegraphics[width=0.8\textwidth]{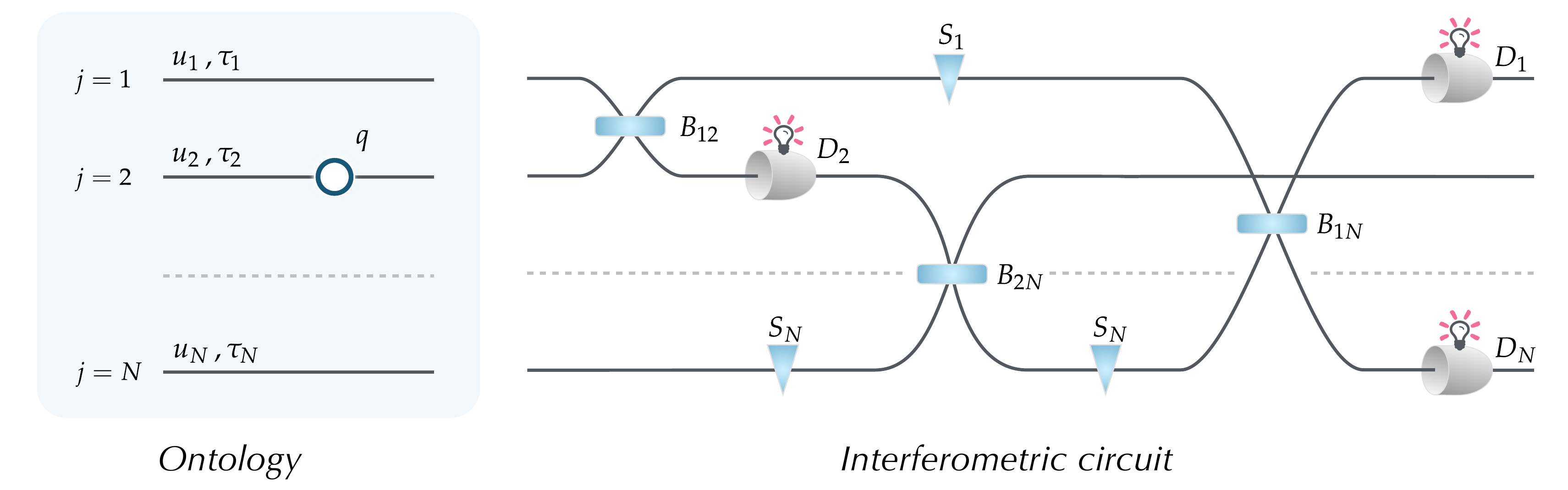}
\end{center}
\caption{\label{Fig-Ontology}{\bf Ontology of the model and interferometric circuits.} One the left, ontology of the model consists of a \emph{single particle} and \emph{local fields} propagating in each path of the circuit. At each time the particle has well-defined \emph{position} $q=1,\dots,N$ and the fields are characterised by \emph{amplitude} $u_j$ and \emph{strength} $\tau_j$ with $j=1,\dots,N$ labelling the paths. On the right, circuits describe propagation of a particle through a network of (spatially separated) paths and gates which represent a sequence of transformations. Basic interferometric toolkit consists of \emph{free evolution} (empty path), \emph{phase shifters} $S_j$, \emph{beam splitters} $B_{st}$ (on which two paths meet) and \emph{detectors} $D_j$ which inform (\textsc{'Click'}/\textsc{'No Click'}) about the presence/absence of a particle in a given path. This selection of gates is general enough to provide physical realisation of any unitary and projective measurement described by quantum formalism in $\mathcal{H}=\mathbb{C}^N$~\cite{ReZeBeBe94}. In this paper we show that the outlined ontology (on the left) completed with appropriately defined local stochastic  gates fully reconstructs quantum mechanical predictions for a single particle in the interferometric circuits.}
\end{figure*}

\ \\\noindent\textbf{Local interferometric gates}.
For such defined ontology we need to define stochastic counterparts of the interferometric gates. Note that in order to obey the locality principle action of the gates should be restricted to the paths they are attached to, i.e. modify degrees of freedom only in the respective paths and the effected transformation being not dependent on the situation (configuration of gates, outcomes or fields) in the other paths. 

We start with description of paths without gates which correspond to  \emph{free evolution}. It will be assumed that the field in such a path is subject to \emph{'natural ageing'}, namely at each step its strength decreases and amplitude remains unchanged. We make the following definition of \emph{free evolution} in $j$-th path:
\begin{eqnarray}\label{Gates-Free}
u_j\stackrel{free}{\longrightarrow} u_j&\ \ \ \text{\&}\ \ \ &\tau_j\stackrel{free}{\longrightarrow}\nicefrac{\tau_j\,}{2}\,.
\end{eqnarray}

\emph{Phase shifter} $S_j$ is a deterministic gate which acts in {$j$-th} path by rotating phase of the field by $e^{i\omega}$ and strength \emph{'ageing naturally'}, i.e. we have:
\begin{eqnarray}\label{Gates-P}
u_j\stackrel{S_{j}}{\longrightarrow} e^{i\omega}u_j&\ \ \ \text{\&}\ \ \ &\tau_j\stackrel{S_{j}}{\longrightarrow}\nicefrac{\tau_j\,}{2}\,,
\end{eqnarray}

\emph{Detector} $D_j$ checks for presence of the particle in $j$-th path (i.e. detector \textsc{'Clicks'} only if $q=j$). Furthermore, we postulate that the detection depending on the result (\textsc{'Click'}/\textsc{'No Click'}) modifies amplitude and strength of the field in $j$-th path in the following way:
\begin{eqnarray}\label{Gates-D}
u_j\stackrel{D_{j}}{\longrightarrow} \left\{\begin{array}{l}1\vspace{0.05cm}\\u_j\end{array}\right.&\ \ \ \text{\&}\ \ \ &\tau_j\longrightarrow\left\{\begin{array}{ll}1&\ \ \ \ \ \ \text{if}\ \ q=j\vspace{0.05cm}\\0&\ \ \ \ \ \ \text{if}\ \ q\neq j\end{array}\right..
\end{eqnarray}

In the above definitions it is implicitly assumed that the particle \emph{can not} jump between the paths. In other words, if the particle happens to be in path $q=j$, then it stays there $q\longrightarrow q$, and otherwise for $q\neq j$ it remains outside $q\longrightarrow q\neq j$. The particle may change its location only at the crossing points, i.e. where the the beam splitters are placed.

\emph{Beam splitter} $B_{st}$ is a gate which brings paths $s$ and $t$ together and implements the following transformation. Amplitude and strength of the fields are modified according to the recipe:
\begin{eqnarray}\label{Gates-B-u}
\begin{pmatrix}u_s\\u_t\end{pmatrix}\stackrel{B_{st}}{\longrightarrow}\begin{pmatrix}u_s'\\u_t'\end{pmatrix}=\begin{pmatrix}i\sqrt{R}&\sqrt{T}\\\sqrt{T}&i\sqrt{R}\end{pmatrix}\!\begin{pmatrix}\delta_{\tau_s\tau^{\text{($st$)}}}&\!\!0\\0&\!\!\delta_{\tau_t\tau^{\text{($st$)}}}\end{pmatrix}\!\begin{pmatrix}u_j\\u_k\end{pmatrix},\ \ \ \ 
\end{eqnarray}
and
\begin{eqnarray}\label{Gates-B-tau}
\!\!\!\!\tau_s\,,\tau_t\stackrel{B_{st}}{\longrightarrow}\,\nicefrac{\tau^{\text{($st$)}}}{2}\,,&\ \text{where}&\tau^{\text{($st$)}}=\max\,\{\tau_s\,,\tau_t\}\,.
\end{eqnarray}
In plain words, the role of $\delta$'s in the diagonal matrix in Eq.\,(\ref{Gates-B-u}) is to suppresses the field with weaker strength so that it does not contribute to the transformed amplitudes at the output. Note that strengths of the outgoing fields are subsequently levelled up to $\nicefrac{\tau^{\text{($st$)}}}{2}$; see Eq.\,(\ref{Gates-B-tau}). Additionally, if the particle happens to be in one of the crossing paths, i.e. $q=s$ or $q=t$, then it may change its position following the probabilistic rule:
\begin{eqnarray}\label{Gates-B-q}
q\stackrel{B_{st}}{\longrightarrow} \Bigg\{\begin{array}{ll}q'=s&\,\text{with probability\ \ }\tfrac{|u_s'|^2}{|u_s'|^2+|u_t'|^2}\,,\vspace{0.1cm}\\q'=t&\,\text{with probability\ \ }\ \tfrac{|u_t'|^2}{|u_s'|^2+|u_t'|^2}\,.\end{array}
\end{eqnarray}
and otherwise, for $q\neq s$ and $q\neq t$, it remains outside.

All gates defined above are \emph{local} (with the interaction between the paths on the beam splitter allowed since it is placed at the crossing point). We also note that transformations effected by free evolution, phase shifters $S_j$ and detectors $D_j$ are \emph{deterministic}, while the beam splitters $B_{st}$ are non-trivial \emph{stochastic} gates.

Observe that the structure of circuits constructed in the model is analogous to those in the quantum interferometric framework. The difference lays in the underlying ontology which in the presented model is given explicitly with locality being built in from the outset. We will show that statistical predictions for any experimental circuit in the model are the same as for its quantum-mechanical counterpart.

\ \\\noindent\textbf{Operational desideratum}.
Imagine agent without any prior knowledge of the model making an effort to understand how it works only by analysing results of experiments that she can perform. Clearly, her conception of the model may diverge from the \emph{'true'} ontology described above, since her choice of gates in constructing experimental circuits is constrained. In the following we are interested in forming minimal account of the model as seen by the agent  avoiding any unfounded interpretational commitments. It is thus appropriate to adopt \emph{operational} approach and restrict attention solely to description of experimental predictions in the circuits built according to the rules of the model.

For this purpose we need to identify what information is actually available to the agent subject to this kind of constraints. The following questions provide guidance in this process:
\begin{itemize}
\item [$(i)$]\emph{Which distributions in $\mathcal{P}(\Lambda)$ can be prepared by the agent with limited resources at hand?}
\end{itemize}
In general, it may be the case that the agent explores only a restricted range of distributions in $\mathcal{P}(\Lambda)$, meaning that some distributions are beyond her reach. Then it is natural to ask:
\begin{itemize}
\item [$(ii)$]\emph{How do these distributions transform under action of the gates in the model?}
\end{itemize}
What remains is to abstract away redundant ontological structure. Here is the key to the operational account:
\begin{itemize}
\item []\emph{\textbf{Operational indifference principle:}
Distributions that are not distinguishable by means available to the agent, that is give the same probabilistic predictions for any conceivable experiment (circuit), are equivalent from the operational point of view.}
\end{itemize}
It allows to discard ontological details which are irrelevant (or inaccessible) to the agent by treating all indistinguishable distributions as a single entity. At this point one should be able to identify the underlying mathematical framework and answer the question:
\begin{itemize}
\item [$(iii)$]\emph{What is the minimal operational account which correctly describes predictions of the model?}
\end{itemize}

In short, we seek for the bare-bone description without preference to any particular interpretation, with the only purpose to provide a tool for prediction of experimental results. Such an account should specify the set of possible operational states which correspond to inequivalent preparation procedures  and provide transformation rules describing evolution in conceivable experimental circuits (including measurement outcomes). In the following, we show how to construct such an operational account of the model which makes no reference to the underlying ontology.

\ \\\noindent\textbf{Reconstruction of quantum predictions}.
Closer analysis of the model reveals significance of special classes of distributions $[\,\vec{z}\,]\subset\mathcal{P}(\Lambda)$ which can be labeled with complex vectors (rays) $\vec{z}\in\mathbb{C}^n$, that is
\begin{eqnarray}\label{Quant-State}
\vec{z}=\sum_{j=1}^N\ z_j\,\bm{e}_j=\left(\begin{smallmatrix}z_1\vspace{-0.1cm}\\\vdots\vspace{0.1cm}\\z_N\end{smallmatrix}\right),
\end{eqnarray}
with normalisation $\Vert\vec{z}\Vert=\sum_j|z_j|^2=1$ and equivalence up to the overall phase.
These classes can be shown to form disjoint family of subsets in $\mathcal{P}(\Lambda)$, i.e. we have
\begin{eqnarray}\label{[z]-disjoint}
[\,\vec{z}\,]\cap[\,\vec{z}\,']\neq\varnothing\ \ &\Leftrightarrow&\ \ \vec{z}=\vec{z}\,'\ \text{(up to phase)}\,.
\end{eqnarray}
See Section Methods for explicit definitions and Fig.~\ref{Fig-Classes} for illustration.

Interest in these very special classes of distributions $[\,\vec{z}\,]\subset\mathcal{P}(\Lambda)$ is due to their behaviour under action of the gates defined in the model. Let us summarise main results relevant for the discussion of the operational account of the model (see Section Methods). Firstly, it can be shown that free evolution, phase shifters, beam splitters  and detectors (with post-selection) act \emph{congruently} on such defined family of classes, i.e. all distributions in a given class are mapped into distributions in some other class $[\,\vec{z}\,]\ni \bm{p}\longrightarrow \bm{p}'\in[\,\vec{z}\,']$; see Theorem~\ref{Theorem}. Secondly, one observes that available preparation procedures make the agent start off with distributions contained in one of the \emph{initial classes} $[\,\bm{e}_1\,]\,,\dots,[\,\bm{e}_N\,]$, where $\bm{e}_j=(0,\dots,1\,,\dots,0)^T$ has single 1 in $j$-th position which indicates position of the particle ('\textsc{Click}') ascertained by the initial preparation; see Eq.\,(\ref{Initial[z]}).

Combining these facts together provides answer to questions $(i)$ and $(ii)$ from the operational desideratum discussed  above. Since the family of classes is closed under available transformation, we infer that the agent with a limited choice of gates at command remains confined in her explorations to a restricted subset of distributions in $\mathcal{P}(\Lambda)$ given by the union of all classes, i.e.
\begin{eqnarray}
\bigcup\Big\{\,[\,\vec{z}\,]\,:\,\vec{z}\in\mathbb{C}^N\,,\,\Vert\vec{z}\Vert=1\,\Big\}\subsetneq\mathcal{P}(\Lambda)\,.
\end{eqnarray}
Note that this set has natural coarse-graining (partitioning) into classes $[\,\vec{z}\,]$ which have the property that action of the gates in the model is concisely described as transformation of the labelling vectors $\vec{z}\longrightarrow\vec{z}\,'$. A crucial observation is that on the level of classes these transformation rules are exactly the same as for the quantum interferometric gates; cf. Eqs.~(\ref{Quant-Free+Pj})--(\ref{Quant-Dj}) and Eqs.~(\ref{Thm-Free})--(\ref{Thm-B}) in Theorem~\ref{Theorem}.

Such a coarse-grained description is just enough for our purposes. This is because distributions in the same class $[\,\vec{z}\,]$ give identical measurement predictions, i.e. probability of a '\textsc{Click}' in detector $D_j$ is equal to $|z_j|^2$. Moreover, since classes transform as a whole there is no way to differentiate by the agent between two distributions in the same class by arranging any complicated circuit from the gates available in the model. This allows to make use of the \emph{operational indifference principle} and observe that all information relevant for predicting behaviour of the system is held by the class itself, that is knowledge of a particular distribution in $[\,\vec{z}\,]$ is redundant. It means that label $\vec{z}$ plays the role of operational state which encodes complete information available to the agent, thereby answering question $(iii)$ from the operational desideratum discussed above. Notice that we get full analogy with the quantum description of interferometric circuits, i.e. we have the same structure (geometry)  of states  which are complex vectors (rays) in $\mathbb{C}^N$ with identical transformation and measurement rules given in Eqs.~(\ref{Quant-Free+Pj})--(\ref{Quant-Dj}) and Eqs.~(\ref{Thm-Free})--(\ref{Thm-B}) respectively. All things considered, both descriptions are equivalent and thus we can identify
\begin{eqnarray}\label{equiv}
\ket{\psi}\ \stackrel{equiv.}{\leftrightsquigarrow}\ [\,\vec{z}\,]&&\ \ \ (\,\text{or}\ \  \vec{\psi}\ \stackrel{equiv.}{\leftrightsquigarrow}\ \vec{z}\,)\,.
\end{eqnarray}
 
In conclusion, operational account of the model boils down to specification of a state given by a complex vector (ray) $\vec{z}\in\mathbb{C}^N$ with the transformation rules and statistics of outcomes ('\textsc{Clicks}') being the same as for the quantum gates. This means that from the perspective of an agent unaware or indifferent to the underlying ontology the behaviour of quantum-interferometric circuits and their counterparts in the presented model are for all practical purposes 
indistinguishable.

\section{Discussion}

In summary, we have constructed local ontological model which faithfully imitates quantum predictions for a single particle in the interferometric circuits. Crucial for the analysis of the model is distinction between two levels of description. On the one hand, we have \emph{ontological description} by an omniscient observer having access to all details of the model, i.e. seeing structure of the ontic state space and familiar with construction of the gates. On the other hand, we have \emph{epistemic description} concerned only with the information which is actually available. The latter adopts \emph{operational perspective} of an agent unaware of the underlying ontology and investigating the system only with the tools at hand, i.e. building interferometric circuits and analysing experimental results (statistics of '\textsc{Clicks}'). We have shown that operational predictions of the constructed model are indistinguishable from the quantum mechanical behaviour. This illustrates that properly chosen constraints on gaining knowledge can modify the picture on the epistemic level. In our case, from the local ontology with classical probabilistic description in $\mathcal{P}(\Lambda)$ we see emergent geometry of the projective space $\mathcal{H}=\mathbb{C}^N$ and quantum mechanical account of a qudit~\cite{BeZy06}.

This result is an explicit counterexample showing impossibility of proving non-locality for a single particle in the interferometric setups. For the sake of clarity, we address the question of (non-)locality in a different context than the Bell-type scenario; the latter is concerned with correlations between measurements on a pair of quantum particles, whereas here we are concerned with a single quantum particle interacting with classical apparatus (phase shifters, beam-splitters and detectors) as described by quantum theory in $\mathcal{H}=\mathbb{C}^N$~\cite{ReZeBeBe94}.

At first sight our conclusion seems to contradict proofs claiming non-locality of a single particle~\cite{TaWaCo91,Ha94,GrHoZe95,DuVe07}. We note that these arguments exploit additional quantum resource, namely coherent states whose properties rely on superposition of multi-particle states. This requires presence of other particles in the system making the claim of single-particle character of the considered phenomena open to question~\cite{GrHoZe95}. A similar objection applies to recent demonstration of the collapse of the wave function using homodyne detection~\cite{FuTaZwWiFu15}. In view of the presented model, these proofs seem to illustrate non-trivial aspect of \emph{'almost'} classical resource provided by local oscillators (coherent states), as compared with \emph{'clean'} single-particle scenarios considered in this paper.

We note that the single-particle framework is a rich source of paradoxes and weird phenomena which are often considered as typically quantum effects without classical explanation~\cite{ScSu98,ElVa93,KwWeHeZeKa95,De02,MiSu77,Pe80,KwWeHeZeKa95,De02,Wh78,JaWuGrTrGrAsRo07,LeGa85,EmLaNo14,AhAlVa88,GeRoMaBlBeMaTw13,KoSp67,KlCaBiSh08}. The latter assertion should be treated with caution, since any argument for non-classicality of an effect always depends on additional assumptions whose plausibility should be properly assessed. For example, interaction-free measurements assume null effect of negative measurement results~\cite{ElVa93,KwWeHeZeKa95,De02}, Leggett-Garg inequalities require non-invasive measurements, pre- and post-selection paradoxes 
rest upon contextual effects~\cite{LeSp05,Ma12,Pu14}, etc. Our model illustrates non-trivial aspect of these assumptions and shows that mere local state disturbance by detectors ushers in a possibility for classical-like explanation of single-particle phenomena. A strong point of the model is that the presented ontology is made ready for any kind of circuit with arbitrary number of paths. As such, it provides exhaustive reconstruction of single-particle phenomena in a unified framework as opposed to separate models devised for simulation of particular effects, cf.~\cite{Ha99,DaPlPl02,Sp07,KlGuPoLaCa11,BaRuSp12,WaBa12,Bl13,FeCo14,Bl15,Bl15a,KaCaBaRu15}.

Let us remark that in the classification of Harrighan-Spekkens~\cite{HaSp10} our construction is $\psi$-ontic, that is distributions corresponding different quantum states have non-overlapping supports. 
We should also point out that the model allows for different representations of the the same quantum state, that is any distribution $\bm{p}\in[\,\vec{z}\,]$ is a valid representation of the same state $\ket{\psi}$ (with the identification $\vec{\psi}\!\leftrightsquigarrow\!\vec{z}\,$). This variety is necessary to accommodate contextual effects which abound in the quantum regime~\cite{KoSp67,KlCaBiSh08,Sp05,HaRu07}.

To give a broader perspective we hasten to note that there is only a handful of ontological models which reconstruct a qudit. One of them is the $\psi$-ontic model by Beltrametti-Bugajski~\cite{BeBu95} which is essentially restatement of the standard Copenhagen interpretation (where non-locality of the collapse of the wave function is built in from the outset). There is also an interesting proposal by Lewis-Jennings-Barrett-Rudolph~\cite{LeJeBaRu12} built within the framework of $\psi$-epistemic models. 
It is explicitly non-local and, in addition, violates the so called preparation independence principle -- the latter seems to be a generic feature of any successful $\psi$-epistemic approach, see~\cite{PuBaRu12,Le14a}. One should also mention the de Broglie-Bohm interpretation of quantum mechanics~\cite{Bo52,CuFiGo96} which postulates local guidance of particles by a quantum potential. For a single particle quantum potential (directly related to the wave function) lives in a 3D space and its dependence on the configuration of the apparatus is a source of non-local effects. Additionally, the de Broglie-Bohm model has many weird features, such as complicated spatial description, 'surrealistic' trajectories~\cite{EnScSuWa92} or excessive contextual effects~\cite{Ha96a}, which persist even in the simple interferometric setups whose relevant degrees of freedom reduce to a qudit. In summary, all these models have built in non-local effects in the description and therefore do not make a case against non-locality of single-particle interferometry discussed 
in this paper.

To conclude, let us quote E. T. Jaynes~\cite{Ja90} on the current understanding of quantum mechanical formalism: \emph{"But our present QM formalism is not purely epistemological; it is a peculiar mixture describing in part realities of Nature, in part incomplete human information about Nature -- all scrambled up by Heisenberg and Bohr into an omelette that nobody has seen how to unscramble. Yet we think that the unscrambling is a prerequisite for any further advance in basic physical theory. For, if we cannot separate the subjective and objective aspects of the formalism, we cannot know what we are talking about; it is just that simple."} In this spirit, our model is an illustration of the idea that careful distinction between the epistemic aspect of the description and the underlying ontological account provides a way of understanding weird quantum phenomena as an effect of incomplete knowledge -- it is tenable at least for single-particle framework as the model demonstrates. This gives support to the belief that \emph{unscrambling of the quantum omelette} should be in principle possible, albeit it is not evident at the moment how to construct such a theory. It seems that non-local effects should play a role in the full reconstruction -- as the Bell's theorem suggests -- however, it is not clear to what extent and in what form (see~\cite{ToBa03} for some hints). The presented model points out to the multi-particle behaviour as the real source of the quantum mystery in comparison with the single-particle phenomena which are less problematic in this respect. In particular, we have shown that single-particle framework is not enough to establish non-locality, since in this case \emph{'spooky action-at-a-disstance'} can be understood as merely an effect of description on the epistemic level.


\section{Methods}

Here we give all necessary definitions and state our main result which describes structure of distributions within the reach of an agent exploring the model (see Supplementary Information for the proof).

\ \\\noindent\textbf{Classes of distributions in $\bm{\mathcal{P}(\Lambda)}$}.
Crucial for the analysis of the model is the following construction of distinguished classes of distributions in $\mathcal{P}(\Lambda)$; see Fig.~\ref{Fig-Classes} for illustration. 
\begin{figure*}[t]
\begin{center}
\includegraphics[width=1\textwidth]{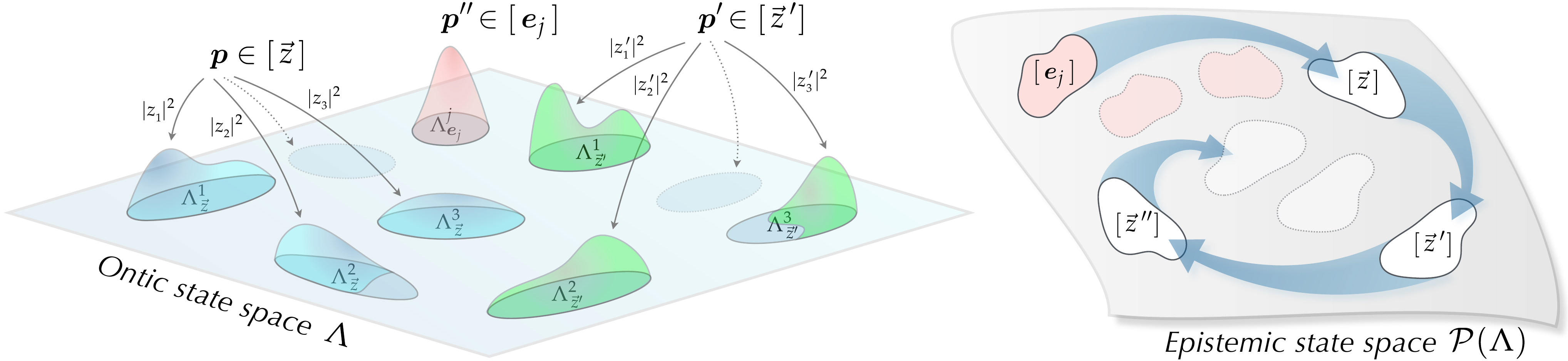}
\end{center}
\caption{\label{Fig-Classes}{\bf Construction of classes of interest in $\mathcal{P}(\Lambda)$.} On the left, distributions $\bm{p}\in[\,\vec{z}\,]$ are defined to have support in $\bigcup_{i=1}^N\Lambda_{\,\vec{z}}^i$ with cumulative probability over the respective subsets $\Lambda_{\,\vec{z}}^i$ equal to  $|z_i|^2$. Since all $\Lambda_{\,\vec{z}}^i$ are disjoint,  distributions in different classes have non-overlapping supports (see $\bm{p}\in[\,\vec{z}\,]$ and $\bm{p}'\in[\,\vec{z}\,']$). 
\vspace{-0.1cm}
Classes $[\,\bm{e}_j\,]$ for $j=1,\dots,N$ are special: each is comprised of distributions supported in a single subset $\Lambda_{\,\bm{e}_j}^j$ (see $\bm{p}''\in[\,\bm{e}_j\,]$) and describes initial preparations with the particle in a given path. On the right, illustration of the whole space of probability distributions over $\Lambda$, denoted by $\mathcal{P}(\Lambda)$,
\vspace{0.05cm}
with disjoint subsets representing classes of interest $[\,\vec{z}\,]$. These classes transform congruently
\vspace{0.05cm}
(as a whole) under action of the gates in the model as explained in Theorem~\ref{Theorem} (in the picture initial class undergoes sequence of transformations $[\,\bm{e}_j\,]\rightarrow[\,\vec{z}\,]\rightarrow[\,\vec{z}\,']\rightarrow[\,\vec{z}\,'']\rightarrow\dots$).}
\end{figure*}

\noindent\textbf{\emph{Step 1:}} Let us consider special subsets of the ontic state space  $\Lambda_{\,\vec{z}}^i\subset\Lambda$ labeled by integers $i\in\{1,\dots,N\}$ and complex vectors $\vec{z}\in\mathbb{C}^N$ defined as follows:
\begin{eqnarray}\label{Lambda_zi}
(q,\vec{u},\vec{\tau})\in\Lambda_{\,\vec{z}}^i&\ \stackrel{df}{\Longleftrightarrow}\ & \left\{
\begin{array}{ll}
a)\ &q=i\vspace{0.1cm}\\
b)\ &\tau_i=\tau>0\vspace{0.1cm}\\
c)\ &\Delta_\tau\,\vec{u}\sim\,\vec{z}
\end{array}\right.
\end{eqnarray}
where $\Delta_\tau\,\vec{u}$ is a vector obtained from $\vec{u}$ by retaining field amplitudes corresponding to the highest field strength $\tau\!:=\max\,\{\tau_1\,,\dots,\tau_N\}$ and the remaining ones put equal to zero. Hence the role of the diagonal matrix
\begin{eqnarray}
\Delta_\tau\!:=\text{diag}\,(\delta_{\tau_1\tau}\,,\dots,\,\delta_{\tau_N\tau})\,,
\end{eqnarray}
which picks out those entries of $\vec{u}$ which correspond to the highest strength $\tau$. In our notation symbol '\!$\sim$' stands for proportionality, i.e. $\vec{z}\sim\vec{z}\,'$ iff $\vec{z}=\alpha\,\vec{z}\,'$ for some $\alpha\in\mathbb{C},\ \alpha\neq0$. In plain words, these conditions express the following requirements:\vspace{-0.1cm}
\begin{itemize}
\item[$a)$]{particle is present in path $i$,\vspace{-0.1cm}}
\item[$b)$]{field in path $i$ has highest strength (non-vanishing; with possibility of equal strengths in other paths),\vspace{-0.1cm}}
\item[$c)$]{vector of field amplitudes with highest strengths $\Delta_\tau\,\vec{u}$ is proportional to $\vec{z}$.\vspace{-0.1cm}}
\end{itemize}
Clearly, for different labels $i$ and $\vec{z}$ (up to proportionality) these subsets are disjoint, i.e. we have
\begin{eqnarray}\label{ij-disjoint}
\Lambda_{\,\vec{z}}^i\,\cap\Lambda_{\,\vec{z}\,'}^j\neq\varnothing\ \ &\ \Leftrightarrow\ &\ \ i= j\ \ \ \&\ \ \ \vec{z}\sim\vec{z}\,'\,.
\end{eqnarray}

\noindent\textbf{\emph{Step 2:}} Then, we introduce auxiliary classes of probability distributions with support in $\Lambda_{\,\vec{z}}^i$ and denote
\begin{eqnarray}\label{[z]i}
[\,\vec{z}\,]_i:=\Big\{\bm{p}\in\mathcal{P}(\Lambda):\,\text{supp}\,\bm{p}\subset\Lambda_{\,\vec{z}}^i\Big\}\subset\mathcal{P}(\Lambda)\,.
\end{eqnarray}
By virtue of Eq.\,(\ref{ij-disjoint}) these classes form a disjoint family of subsets in $\mathcal{P}(\Lambda)$, i.e. we have
\begin{eqnarray}\label{[z]i-disjoint}
[\,\vec{z}\,]_i\cap[\,\vec{z}\,']_j\neq\varnothing\ \ &\ \Leftrightarrow\ &\ \ i= j\ \ \ \&\ \ \ \vec{z}\sim\vec{z}\,'\,.\end{eqnarray}

\noindent\textbf{\emph{Step 3:}} Now, we are ready to define classes of distributions which play a central role in analysis of the model.
\begin{definition}\label{Definition}
With each normalised vector $\vec{z}\in\mathbb{C}^N$, such that $\Vert\vec{z}\Vert:=\sum_{i=1}^N|z_i|^2=1$, we associate the following class of probability distributions:
\begin{eqnarray}\label{[z]}
[\,\vec{z}\,]:=\Big\{\ \sum_{i=1}^N|z_i|^2\,\bm{p}_i:\,\bm{p}_i\in[\,\vec{z}\,]_i\ \Big\}\subset\mathcal{P}(\Lambda)\,.
\end{eqnarray}
\end{definition}
\noindent See Fig.~\ref{Fig-Classes} for illustration. This means that distributions in $[\,\vec{z}\,]$ have support in $\bigcup_{i=1}^N\Lambda_{\,\vec{z}}^i$ with cumulative probability over the respective subsets $\Lambda_{\,\vec{z}}^i$ equal to $|z_i|^2$ (and otherwise the shape of distributions being arbitrary). Another way to characterise classes of interest is to write $[\,\vec{z}\,]=\sum_{i=1}^N|z_i|^2\,[\,\vec{z}\,]_i$\,, 
which means that its elements are convex combinations of distributions in $[\,\vec{z}\,]_i$'s with weights $|z_i|^2$. As a consequence of Eq.\,(\ref{[z]i-disjoint}) we observe that such defined classes are disjoint subsets in $\mathcal{P}(\Lambda)$, i.e. we have
\begin{eqnarray}\label{[z]-disjoint}
[\,\vec{z}\,]\cap[\,\vec{z}\,']\neq\varnothing\ \ &\ \Leftrightarrow\ &\ \ \vec{z}\sim\vec{z}\,'\,.
\end{eqnarray}

\ \\\noindent\textbf{Initial preparation}. 
Any prediction of experimental behaviour rests upon knowledge of initial preparation of the system. In general, it is an intrinsic characteristic of the source which provides an ensemble of systems with a given distribution of the ontic states. However, if no such information is available, then the agent given some unknown (possibly random) source has to prepare initial ensembles of systems by herself. Here is a generic scheme how to proceed in such a case.

Since we are interested in single-particle scenarios, in the first place presence of a single particle ('\textsc{Click}') in the system should be verified. This property can be confirmed by sieving an unknown ensemble through the array of detectors $D_1\,,D_2\,,\dots,D_N$ placed in each path and retaining only those cases when a single detection occurred. In this way the agent carries out an effective initial preparation which attests to the presence of a single particle ('\textsc{Click}') in a given path. Note that on the ontological level selection of events with a single '\textsc{Click}' in detector $D_j$ results in an ensemble distributed over the ontic states  $(q,\vec{u},\vec{\tau})\in\Lambda$ subject to the following conditions: 
\begin{eqnarray}\label{InitialOntic}
\begin{array}{lllllll}q=j&,&u_j=1&,&\tau_j=1&,&\vspace{0.1cm}\\
&&u_k=\text{?}&,&\tau_{\,k}=0&,&\ \ \text{for}\ \ k\neq j\ ,
\end{array}
\end{eqnarray}
where $u_k$'s depend on the unknown source; see Eq.\,(\ref{Gates-D}).
\vspace{-0.08cm}
A quick look at definitions in Eqs.~(\ref{Lambda_zi}), (\ref{[z]i}) and (\ref{[z]}) reveals that such distributions have support in $\Lambda_{\,\bm{e}_j}^j$,
\vspace{-0.05cm}
 and hence are included in class
\begin{eqnarray}\label{Initial[z]}
[\,\bm{e}_j\,]\subset\mathcal{P}(\Lambda)&\ &\text{(\,if \,$D_j$ \,'\textsc{Clicks}'\,)}\,,
\end{eqnarray}
where $\bm{e}_j=(0,\dots,1,\dots,0)$ has single 1 in $j$-th position. In conclusion, the agent starts off in one of the classes $[\,\bm{e}_1\,]\,,\dots,[\,\bm{e}_N\,]$ which correspond to initial preparation of the system with a single particle ('\textsc{Click}') in a given path.

In the above we have assumed no prior knowledge of the source and hence the need of initial filtering of the unknown ensemble. We note that it could have been bypassed if the agent was granted access to a single particle source with all paths blocked except one (like it is usually assumed in the quantum scenarios). This can be easily realised within the model by postulating that the source injects particles (with non-vanishing amplitudes and strengths) into a given path, and the blocks remove particles resetting strength of the field to zero. We observe that it boils down to preparation of distributions in one of the classes in Eq.\,(\ref{Initial[z]}) again.

\ \\\noindent\textbf{Geometry of classes}.
It appears that the structure of classes $[\,\vec{z}\,]\subset\mathcal{P}(\Lambda)$ defined in Eq.\,(\ref{[z]}) is closed under transformations (circuits) considered in the model. Here is the key result describing behaviour of classes under action of the gates in the model, cf. Fig.~\ref{Fig-Classes} (on the right). (For the proof see Supplementary Information).

\begin{theorem}\label{Theorem}
Transformations implemented by phase shifters $S_j$, detectors $D_j$ and beam splitters $B_{st}$  act congruently on the family of classes $\big\{[\,\vec{z}\,]\subset\mathcal{P}(\Lambda):\vec{z}\in\mathbb{C}^n,\Vert\vec{z}\Vert=1\big\}$ defined in Eq.\,(\ref{[z]}). This means that classes transform as a whole, i.e.  all distributions in a given class map into distributions in some other class
\begin{eqnarray}
[\,\vec{z}\,]\ni\bm{p}\ \xymatrix{\ar[r] &}\ \bm{p}'\in[\,\vec{z}\,'\,]\,,
\end{eqnarray}
where mapping $\vec{z}\longrightarrow\vec{z}\,'$ is determined by the gates implemented in the circuit according to the following rules.\vspace{0.15cm}\\
$\bullet$ Free evolution acts trivially and phase shifter $S_j$ introduces phase in the relevant component of vector $\vec{z}$
\begin{eqnarray}
\label{Thm-Free}
z_j\ \xymatrix{\ar[r]^{free} &}\ z_j&\ \ \ \text{and}\ \ \ &z_j\ \xymatrix{\ar[r]^{S_j} &}\ e^{i\omega}\,z_j\,.
\end{eqnarray}
\vspace{0.15cm}
$\bullet$ Detector $D_j$ placed in $j$-th path '\textsc{Clicks}' with probability
\begin{eqnarray}
\label{Thm-D-prob}
\textsf{Pr}\,(D_j|\vec{z})=|z_j|^2\,,
\end{eqnarray}
and depending on the outcome effects projection of vector $\vec{z}$
\begin{eqnarray}
\label{Thm-D}
\vec{z}\ \xymatrix{\ar[r]^{D_j} &}\left\{
\begin{array}{cl}
\bm{e}_j&\ \ $'\textsc{Click}'$\,,\vspace{0.15cm}\\
\tfrac{(\mathbb{1}-\mathbb{P}_j)\,\vec{z}}{\Vert(\mathbb{1}-\mathbb{P}_j)\,\vec{z}\,\Vert}&\ \ $'\textsc{No Click}'$\,.
\end{array}
\right.
\end{eqnarray}
$\bullet$  Beam splitter $B_{st}$ at the crossing of two paths $s$ and $t$ implements the following unitary on the corresponding components of vector $\vec{z}$
\begin{eqnarray}
\label{Thm-B}
\begin{pmatrix}z_s\\z_t\end{pmatrix}\ \xymatrix{\ar[r]^{B_{st}} &}\ \begin{pmatrix}i\sqrt{R}&\sqrt{T}\\\sqrt{T}&i\sqrt{R}\end{pmatrix}\!\begin{pmatrix}z_s\\z_t\end{pmatrix}.
\end{eqnarray}
\end{theorem}
\noindent We note that Theorem \ref{Theorem} is also valid for parallel transformations (gates) implemented in different paths at the same time. 
In such a case, it is implied that evolution is given by joint transformation of the respective components of vector $\vec{z}$ (see Theorem~\ref{Theorem-Matrix} in Supplementary Information for detailed formulation). 

\bibliography{CombQuant}

\newpage\ \newpage


\onecolumngrid
\begin{center}
{\large\bf{Is single-particle interference spooky?}}\vspace{0.3cm}\\
{\large\emph{Supplementary Information\vspace{0.1cm}\\Proof of Theorem~\ref{Theorem}}}\vspace{0.4cm}\\
Pawel Blasiak\\
\emph{Institute of Nuclear Physics Polish Academy of Sciences, PL-31342 Krak\'ow, Poland}
\end{center}\vspace{0.1cm}\hfill
\twocolumngrid

\section*{Preliminaries}
For the proof we switch to the matrix notation and denote $\vec{z}\in\mathbb{C}^N$ as a column vector
\begin{eqnarray}\label{z-State}
\vec{z}=\sum_{j=1}^N\ z_j\,\bm{e}_j=\left(\begin{smallmatrix} z_1\vspace{-0.1cm}\\\vdots\vspace{0.1cm}\\ z_N\end{smallmatrix}\right)\,.
\end{eqnarray}
Let us write out explicitly matrices relevant for the following analysis. We will use diagonal matrices
\begin{eqnarray}\label{Unitary-P}
\!\!\mathbb{P}_j=\left(\begin{smallmatrix} 0&&&&\vspace{-0.25cm}\\
&\ddots&&&\\
&&1&&\vspace{-0.2cm}\\
&&&\ddots&\\
&&&&0_{\,\!}\end{smallmatrix}\right)&\!\!,\ \ &\mathbb{S}_k=\left(\begin{smallmatrix} 1&&&&\vspace{-0.25cm}\\
&\ddots&&&\\
&&e^{i\omega}\!\!&&\vspace{-0.2cm}\\
&&&\ddots&\\
&&&&1\end{smallmatrix}\right),
\end{eqnarray}
where $\mathbb{P}_j$ is a projector on $j$-th
\vspace{-0.05cm}
 component and $\mathbb{S}_k$ introduces phase $e^{i\omega}$ in $k$-th component leaving the remaining ones unchanged. Another useful matrix which acts nontrivially only in components $\{s,t\}$ has the form
\begin{eqnarray}\label{Unitary-Bst}
\mathbb{B}_{st}=
\left(\begin{smallmatrix}
1&&&&&&\vspace{-0.25cm}\\\vspace{0cm}
&\ddots&&&\\
&&\hspace{-0.2cm}\ i\sqrt{R}&.\ .\ .&\sqrt{T}&\\\vspace{0.1cm}
&&\hspace{-0.2cm}\vdots&\ddots&\vdots&\\\vspace{-0.1cm}
&&\hspace{-0.2cm}\sqrt{T}&.\ .\ .&\ i\sqrt{R}&\vspace{0cm}\\
&&&&&\hspace{-0.2cm}\ddots&\\
&&&&&&1
\end{smallmatrix}\right),
\end{eqnarray}
where $R$ and $T$ are some constants such that $R+T=1$. Clearly, matrices $\mathbb{S}_j$ and $\mathbb{B}_{st}$ are unitary.

In the following we consider evolution of the system implemented by a parallel configuration of gates acting in \emph{different} paths at the same time -- this corresponds to a single step in the circuit model. It will be convenient to group paths with the same kind of gates and define the corresponding subsets as
\begin{eqnarray}\label{Partition}
\begin{array}{lll}
\mathcal{F}&-&\text{paths without gates (empty paths),}\vspace{0.05cm}\\
\mathcal{D}&-&\text{paths with detectors,}\vspace{0.05cm}\\
\mathcal{S}&-&\text{paths with phase shifters,}\vspace{0.05cm}\\
\mathcal{B}&-&\text{pairs of paths crossing on beam splitters.}
\end{array}
\end{eqnarray}
Note that $\mathcal{F}$, $\mathcal{D}$, $\mathcal{S}$ and $\bigcup\mathcal{B}$ are disjoint and exhaustive; hence it is a partition of the set of all paths $\{1,\dots,N\}$.\footnote{We have $\mathcal{F},\mathcal{D},\mathcal{S}\subset\{1,\dots,N\}$, and since $\mathcal{B}$ is a set of (unordered) pairs $\mathcal{B}=\big\{\{s_1,t_1\},\{s_2,t_2\},\dots\big\}$ we get $\bigcup\mathcal{B}=\{s_1,t_1,s_2,t_2,\dots\}\subset\{1,\dots,N\}$. By exhaustive we mean $\mathcal{F}\cup\mathcal{D}\cup\mathcal{S}\cup\bigcup\mathcal{B}=\{1,\dots,N\}$. These sets are disjoint because gates act in \emph{different} paths.}

Now, we give a detailed version of Theorem~\ref{Theorem} which in the matrix notation takes the following form.

\begin{theorem}\label{Theorem-Matrix}
Parallel configuration of gates acts congruently on the family of classes $\big\{[\,\vec{z}\,]\!\subset\!\mathcal{P}(\Lambda)\!\!:\vec{z}\in\mathbb{C}^N,\Vert\vec{z}\Vert\!=\!1\big\}$ defined in Eq.\,(\ref{[z]}). This means that classes transform as a whole, i.e. all distributions in a given class map into distributions in some other class
\begin{eqnarray}
[\,\vec{z}\,]\ni\bm{p}\ \xymatrix{\ar[r] &}\ \bm{p}'\in[\,\vec{z}\,'\,]\,,
\end{eqnarray}
where mapping $\vec{z}\longrightarrow\vec{z}\,'$ depends on the configuration of gates $\mathcal{F},\mathcal{D},\mathcal{S},\mathcal{B}$ and measurement outcomes ('\textsc{Clicks}'). It is specified by the following rules:\vspace{0.15cm}\\
$\bullet$ For the system described by a distribution in class $[\,\vec{z}\,]$ detector $D_j$ '\textsc{Clicks}' with probability
\begin{eqnarray}\label{Theorem-Matrix-Prob}
\textsf{Pr}\,(D_j|\,\vec{z}\,)=\ |z_j|^2\,,
\end{eqnarray}
and conditioning (post-selecting) on a '\textsc{Click}' in $D_j$ leaves the system in state described by a distribution $\bm{p}'\in[\,\bm{e}_j\,]$, i.e.
\begin{eqnarray}\label{Theorem-Matrix-Dj}
\vec{z}\ \xymatrix{\ar[r]^{D_j} &}\ \vec{z}\,'=\,\bm{e}_j&\ \ &('\textsc{Click}')\,.
\end{eqnarray}
At each run of experiment either one of the detectors '\textsc{Click}' or all detectors remain silent (negative measurement result), with the latter happening with probability $1-\sum_{j\in\mathcal{D}}\ |z_j|^2$.\vspace{0.15cm}\\
$\bullet$ In the case of negative measurement result ('\textsc{No Click}' in all detectors $D_j$ for $j\in\mathcal{D}$) or no measurement at all (no detectors $\mathcal{D}=\varnothing$) transformation implemented by the gates is given by
\begin{eqnarray}\label{Theorem-Matrix-Negative}
\vec{z}\ \xymatrix{\ar[r] &}\ \vec{z}\,'\sim\prod_{j\in\mathcal{D}}(\mathbb{1}-\mathbb{P}_j)\prod_{k\in\mathcal{S}}\mathbb{S}_k\!\prod_{\{s,t\}\in\mathcal{B}}\mathbb{B}_{st}\ \,\vec{z}\ ,\ \ \ \ 
\end{eqnarray}
with the order of matrices in the product being irrelevant.\footnote{\label{Footnote-Norm}Due to non-unitary projections length of $\vec{z}\,'$ in Eq.\,(\ref{Theorem-Matrix-Negative}) may be less than 1. Hence the proportionality symbol '\!$\sim$' expressing the need of subsequent renormalisation $\vec{z}\,'\leadsto\nicefrac{\vec{z}\,'}{\Vert\vec{z}\,'\Vert}$ (cf. analogous issue in the description of quantum measurement in Eq.\,(\ref{Quant-Dj})).}
\end{theorem}
This formulation deals explicitly with the case of parallel transformations in different paths. It is straightforward to convince oneself that Theorem~\ref{Theorem} follows from Theorem~\ref{Theorem-Matrix} (both are actually equivalent with the latter one being a more rigorous version). It is thus enough to prove the matrix version given above. 

\section{A helpful Lemma}

We begin with a lemma describing transformation of distributions in auxiliary classes $[\,\vec{z}\,]_i\subset\mathcal{P}(\Lambda)$ defined in Eq.\,(\ref{[z]i}) implemented by a single step in the circuit.
\begin{lemma}\label{Lemma}
Suppose that the system is described by a distribution $\bm{p}\in[\,\vec{z}\,]_i$. Then, parallel configuration of gates specified by $\mathcal{F},\mathcal{D},\mathcal{S},\mathcal{B}$ effects transformation with the following properties.
\vspace{0.15cm}\\
$\bullet$ If there is a detector placed in $i$-th path (i.e. we have $i\in\mathcal{D}$), then $D_i$  '\textsc{Clicks}' with certainty and afterwards the system is described by a distribution $\bm{p}'\in[\,\bm{e}_i\,]_i$, i.e. we have
\begin{eqnarray}\label{Lemma-Di}
[\,\vec{z}\,]_i\ni\bm{p}\ \xymatrix{\ar[r]^{D_i} &}\ \bm{p}'\in[\,\bm{e}_i\,]_i&\ \ &('\textsc{Click}')\,,
\end{eqnarray}
and all other detectors $D_j$ with $j\neq i$ remain silent.
\vspace{0.15cm}\\
$\bullet$ If there is no detector in $i$-th path (i.e. we have $i\notin\mathcal{D}$), then none of the detectors '\textsc{Click}' and afterwards the system is described by a distribution $\bm{p}'$ characterised by vector$\,^{\ref{Footnote-Norm}}$
\begin{eqnarray}\label{lemma-z'}
\vec{z}\,'\sim\prod_{j\in\mathcal{D}}(\mathbb{1}-\mathbb{P}_j)\prod_{k\in\mathcal{S}}\mathbb{S}_k\!\prod_{\{s,t\}\in\mathcal{B}}\!\mathbb{B}_{st}\ \,\vec{z}\,,
\end{eqnarray}
according with the following rules:\vspace{0.15cm}

$\diamond$ In the case when $i$-th path does not go into a beam splitter (i.e. we have $i\notin\bigcup\mathcal{B}$), then 
\begin{eqnarray}\label{zz}
[\,\vec{z}\,]_i\ni\bm{p}\ \xymatrix{\ar[r] &}\ \bm{p}'\in[\,\vec{z}\,'\,]_i\,.
\end{eqnarray}

$\diamond$  Otherwise, if $i$-th path goes into the beam splitter $\mathbb{B}_{st}$ (i.e. we have $i\in\{s,t\}\in\mathcal{B}$), then we get
\begin{eqnarray}\label{zz-BS}
[\,\vec{z}\,]_i\ni\bm{p}\ \xymatrix{\ar[r] &}\ 
\bm{p}'=\tfrac{|z_s'|^2}{|z_s'|^2+|z_t'|^2}\,\bm{p}_s'+\tfrac{|z_t'|^2}{|z_s'|^2+|z_t'|^2}\,\bm{p}_t'\,,\ \ \ \ 
\end{eqnarray}
with distributions $\bm{p}_s'\in[\,\vec{z}\,']_s$ and $\bm{p}_t'\in[\,\vec{z}\,']_t$ (it holds for $|z_s'|^2+|z_t'|^2\neq0$).
\end{lemma}

\begin{proof}[\textbf{Proof of Lemma~\ref{Lemma}}]\ \vspace{0.1cm}

Let us take $\bm{p}\in[\,\vec{z}\,]_i$. From the definition of Eq.\,(\ref{[z]i}) we have that
the ontic state $(q,\vec{u},\vec{\tau})$ of the system described by distribution $\bm{p}$ is certainly in the subset $\Lambda_{\vec{z}}^i$\,, meaning that it satisfies three conditions of Eq.\,(\ref{Lambda_zi})
\begin{eqnarray}\label{abc}
\!\!a)\ \ q=i\,,\ \ \ &
b)\ \ \tau_i=\tau>0\,,\ \ \ \ &
c)\ \ \Delta_\tau\,\vec{u}\sim\,\vec{z}\,,
\end{eqnarray}
where
\begin{eqnarray}\label{abc-tau}
\tau&\!:=\!&\max\,\{\tau_1\,,\dots,\tau_N\}\,,\\\label{abc-delta}
\Delta_\tau&\!:=\!&\text{diag}\,(\delta_{\tau_1\tau}\,,\dots,\,\delta_{\tau_N\tau})\,.
\end{eqnarray}

In the following we seek the form of distribution $\bm{p}'$ obtained from $\bm{p}$ as a result of parallel configuration of gates specified by $\mathcal{F}$, $\mathcal{D}$, $\mathcal{S}$ and $\mathcal{B}$. Our strategy is to take an ontic state in support of $\bm{p}$, i.e. satisfying conditions in Eq.\,(\ref{abc}), and check its properties after the transformation. This will give knowledge about the support of distribution $\bm{p}'$ and then comparison with conditions in Eqs.~(\ref{Lambda_zi}) and (\ref{[z]i}) will prove the result.

$\bullet$ {\it First part -- Eq.\,(\ref{Lemma-Di}).}\ \vspace{0.1cm}

First, we note that $\textit{\textsf{Pr}}\,(q=i)=1$ and $\textit{\textsf{Pr}}\,(q\neq i)=0$. It means that detector $D_i$ placed in $i$-th path '\textsc{Clicks}' with certainty and detectors in other paths $D_j$ with $j\neq i$ remain silent ('\textsc{No Click}'). Moreover, after the detection the particle remains in the path
\begin{eqnarray}
q\longrightarrow q'=q=i\,,
\end{eqnarray}

Second, along with a '\textsc{Click}' detector $D_i$ modifies amplitude and strength of the field in $i$-th path as described by Eq.\,(\ref{Gates-D}), i.e. we get
\begin{eqnarray}
u_i\longrightarrow u_i'=1&\ \ \text{and}\ \ &\tau_i\longrightarrow\tau_i'=1\,.
\end{eqnarray}

Third, a quick look at Eqs.~(\ref{Gates-Free})-(\ref{Gates-B-tau}) reveals that strength of the fields in other paths decreases, which entails that $\tau_m\longrightarrow\tau_m'<1$ for $m\neq i$. Together with the previous equation it gives
\begin{eqnarray}
\tau':=\max\,\{\tau_1'\,,\dots,\tau_n'\}=1=\tau_i'\,,
\end{eqnarray}
and hence we obtain $\Delta_{\tau'}=\text{diag}\,(0\,,\dots,1\,,\dots\,0)$ with a single 1 in $i$-th place. Therefore, we have the identity
\begin{eqnarray}
\Delta_{\tau'\,}\vec{u}'=\bm{e}_i\,.
\end{eqnarray}

Putting all this together, we infer that for any configuration of gates with $D_i$ in $i$-th path (i.e. for $i\in\mathcal{D}$) after the transformation  the system is left in the ontic state $(q',\vec{u}',\vec{\tau}')$ which satisfies the conditions 
\begin{eqnarray}
\!\!\!\!\!\!\!\!a)\ \ q'=i\,,\ \ \ &
b)\ \ \tau_{i}'=\tau'>0\,,\ \ \ \ &
c)\ \ \Delta_{\tau'\,}\vec{u}'\sim\,\bm{e}_i\,.
\end{eqnarray}
In consequence, any distribution $\bm{p}\in[\,\vec{z}\,]_i$ gets transformed into a distribution $\bm{p}'\in[\,\bm{e}_i\,]_i$ (see Eqs.~(\ref{Lambda_zi}) and (\ref{[z]i})). This proves the first part of Lemma~\ref{Lemma} Eq.\,(\ref{Lemma-Di}).\vspace{0.2cm} 

$\bullet$ {\it Second part -- Eqs.~(\ref{lemma-z'})-(\ref{zz-BS}).}\ \vspace{0.1cm}

We look at the second part of Lemma~\ref{Lemma} when there is no detector in $i$-th path (i.e. $i\notin\mathcal{D}$). Clearly, in this situation all detectors remain silent ('\textsc{No Click}'), since the particle is in $i$-th path (i.e. $q=i\notin\mathcal{D}$). 

Let us begin by writing explicitly how strength of the field changes in each path. From Eqs.~(\ref{Gates-Free})-(\ref{Gates-B-tau}) we have:
\begin{eqnarray}\label{F_prime}
&&\tau_l\longrightarrow\tau_l'=\nicefrac{\tau_l}{2}\ \ \ \ \ \ \ \ \,\text{for}\ \ \ l\in\mathcal{F}\,,\\\label{S_prime}
&&\tau_k\longrightarrow\tau_k'=\nicefrac{\tau_k}{2}\ \ \ \ \ \ \ \ \text{for}\ \ \ k\in\mathcal{S}\,,\\\label{D_prime}
&&\tau_j\longrightarrow\tau_j'=0\ \ \ \ \ \ \ \ \ \ \ \ \text{for}\ \ \ j\in\mathcal{D}\,,\\\label{B_prime}
&&\tau_r\longrightarrow\tau_r'=\nicefrac{\tau^{\text{($st$)}}}{2}\ \ \ \ \ \ \text{for}\ \ \ r\in\{s,t\}\in\mathcal{B}\,,
\end{eqnarray}
where $\tau^{\text{($st$)}}:=\max\,\{\tau_s,\tau_t\}$. Together with the defining condition $\tau_i=\tau$ (see Eqs.~(\ref{abc}) and (\ref{abc-tau})) it entails that $\tau_m'\leqslant \nicefrac{\tau}{2}$ for all paths $m=1,\dots,N$. Furthermore, it follows that $\tau_i'=\nicefrac{\tau_i}{2}$ (since $i\notin\mathcal{D}$ and in case $i\in\{r,s\}\in\mathcal{B}$ we have $\tau^{\text{($st$)}}=\tau_i$). Therefore, we get
\begin{eqnarray}\label{T_prime}
\tau'\!:=\max\,\{\tau_1'\,,\dots,\tau_N'\}=\tau_i'=\nicefrac{\tau}{2}>0\,,
\end{eqnarray}
which together with Eqs.~(\ref{F_prime})\,-\,(\ref{B_prime}) allows to compare $\Delta_{\tau'}\!:=\!\text{diag}\,(\delta_{\tau_1'\tau'}\,,\dots,\,\delta_{\tau_N'\tau'})$ with $\Delta_{\tau}$ in Eq.\,(\ref{abc-delta}).

In the following step we will investigate transformation properties of the vector of field amplitudes $\vec{u}$. Since action of each gate is limited to the path(s) it is attached to, the effect of each separate gate can be written in the following way (see Eqs.~(\ref{Gates-Free}) and (\ref{Gates-B-tau}))
\begin{eqnarray}
&&\vec{u}\ \xymatrix{\ar[r]^{free} &}\ \vec{u}'=\vec{u}\,,\\
&&\vec{u}\ \xymatrix{\ar[r]^{S_k} &}\ \vec{u}'=\mathbb{S}_k\,\vec{u}\,,\\\label{D-u}
&&\vec{u}\ \xymatrix{\ar[r]^{D_j} &}\ \vec{u}'=\vec{u}\,,\\
&&\vec{u}\ \xymatrix{\ar[r]^{B_{st}} &}\ \vec{u}'=\mathbb{B}_{st}\,\Delta^{\text{($st$)}}_\tau\,\vec{u}\,,
\end{eqnarray}
where 
\begin{eqnarray}
\!\!\!\!\!\Delta^{\text{($st$)}}_\tau:=\text{diag
}\,(1,\dots,\,\delta_{\tau_s\tau^{\text{($st$)}}}\,,\dots,\,\delta_{\tau_t\tau^{\text{($st$)}}}\,,\dots,1)\,,
\end{eqnarray}
with all 1's on the diagonal except entries $s$ and $t$ which depend on $\tau^{\text{($st$)}}\!:=\max\,\{\tau_s\,,\tau_t\}$. Recall that we consider the case $q=i\notin\mathcal{D}$, and hence for all $j\in\mathcal{D}$ we have $q\neq j$, which explains trivial action of detectors in Eq.\,(\ref{D-u}). Taking all this together, joint transformation implemented by a parallel configuration of gates $\mathcal{F}$, $\mathcal{D}$, $\mathcal{S}$ and $\mathcal{B}$ boils down to the product
\begin{eqnarray}\label{u'}
\vec{u}\ \longrightarrow\ \vec{u}'=\prod_{k\in\mathcal{S}}\mathbb{S}_k\!\prod_{\{s,t\}\in\mathcal{B}}\!\mathbb{B}_{st}\,\Delta^{\text{($st$)}}_\tau\,\vec{u}\,.
\end{eqnarray}

Now, we will justify the following matrix identity:
\begin{eqnarray}\label{Delta-commutation}
\begin{aligned}
\underbrace{\Delta_{\tau'}\left(\prod_{k\in\mathcal{S}}\mathbb{S}_k\!\prod_{\{s,t\}\in\mathcal{B}}\!\mathbb{B}_{st}\,\Delta^{\text{($st$)}}_\tau\right)}_\mathbb{L}\ \ \ \ \ \ \ \ \ \ \ \ \ \ \ \ \ \ \\
=\ \underbrace{\left(\prod_{j\in\mathcal{D}}(\mathbb{1}-\mathbb{P}_j)\prod_{k\in\mathcal{S}}\mathbb{S}_k\!\prod_{\{s,t\}\in\mathcal{B}}\!\mathbb{B}_{st}\right)\Delta_{\tau}}_\mathbb{R}\,.
\end{aligned}
\end{eqnarray}
For the proof we observe that on both sides all matrices in the products are diagonal except for matrices $\mathbb{B}_{st}$ with $2\times2$ blocks acting in entries $\{s,t\}$ (without overlap for different $\mathbb{B}_{st}$). Therefore we have the same block diagonal structure of the matrix both on the left $\mathbb{L}$ and on the right $\mathbb{R}$ hand side, which consists of $\text{1}\times\text{1}$ blocks in entries $l\in\mathcal{F}$, $j\in\mathcal{D}$, $k\in\mathcal{S}$ and $\text{2}\times\text{2}$ blocks in entries $\{s,t\}\in\mathcal{B}$. It is thus enough to verify each block separately in the identity. 
For $\text{1}\times\text{1}$ blocks, we have
\begin{eqnarray}
&&\!\!\!\!\!\!\!\!\!\!\!\!\!\!\!\!\!\!\!\!\!\!\mathbb{L}_{ll}=\delta_{\tau_l'\tau'}\stackrel{(\ref{F_prime})(\ref{T_prime})}{=}\delta_{\tau_l\tau}=\mathbb{R}_{ll}\,,\\
&&\!\!\!\!\!\!\!\!\!\!\!\!\!\!\!\!\!\!\!\!\!\!\mathbb{L}_{kk}=\delta_{\tau_k'\tau'}\,(\mathbb{S}_k)_{kk}\stackrel{(\ref{S_prime})(\ref{T_prime})}{=}(\mathbb{S}_k)_{kk}\,\delta_{\tau_k\tau}=\mathbb{R}_{kk}\,,\\
&&\!\!\!\!\!\!\!\!\!\!\!\!\!\!\!\!\!\!\!\!\!\!\mathbb{L}_{jj}=\delta_{\tau_j'\tau'}\stackrel{(\ref{D_prime})(\ref{T_prime})}{=}0\stackrel{(\ref{Unitary-P})}{=}(1-\mathbb{P}_j)_{jj}\,\delta_{\tau_j\tau}=\mathbb{R}_{jj}\,,
\end{eqnarray}
For $2\times2$ blocks, in the subspace restricted to the respective entries $\{r,s\}\in\mathcal{B}$ we get 
\begin{eqnarray}
&&\!\!\!\!\!\!\!\!\!\!\!\!\!\!\!\!\!\!\!\!\!\!\mathbb{L}_{\{s,t\}}=\begin{pmatrix}\delta_{\tau_s'\tau'}&\!\!0\\0&\!\!\delta_{\tau_t'\tau'}\end{pmatrix}\!\begin{pmatrix}i\sqrt{R}&\sqrt{T}\\\sqrt{T}&i\sqrt{R}\end{pmatrix}\!\begin{pmatrix}\delta_{\tau_s\tau^{\text{($st$)}}}&\!\!0\\0&\!\!\delta_{\tau_t\tau^{\text{($st$)}}}\end{pmatrix},
\end{eqnarray}
and
\begin{eqnarray}
&&\!\!\!\!\!\!\!\!\!\!\!\!\!\!\!\!\!\!\!\!\!\!\mathbb{R}_{\{s,t\}}=\begin{pmatrix}i\sqrt{R}&\sqrt{T}\\\sqrt{T}&i\sqrt{R}\end{pmatrix}\!\begin{pmatrix}\delta_{\tau_s\tau}&\!\!0\\0&\!\!\delta_{\tau_t\tau}\end{pmatrix}.
\end{eqnarray}
In order to show $\mathbb{L}_{\{s,t\}}=\mathbb{R}_{\{s,t\}}$ we use Eqs.~(\ref{B_prime}) and (\ref{T_prime}) to check the following three cases.\\
Case $\tau_s\,,\tau_t<\tau$. Then we have $\tau_s'\,=\tau_t'=\nicefrac{\tau^{\text{($st$)}}}{2}<\nicefrac{\tau}{2}=\tau'$, and hence
\begin{eqnarray}
\mathbb{L}_{\{s,t\}}=0=\mathbb{R}_{\{s,t\}}\,.
\end{eqnarray}
Case $\tau_s=\tau_t=\tau$. Then we have $\tau_s'\,=\tau_t'=\nicefrac{\tau^{\text{($st$)}}}{2}=\nicefrac{\tau}{2}=\tau'$, which gives
\begin{eqnarray}
\mathbb{L}_{\{s,t\}}=\left(\begin{smallmatrix}i\sqrt{R}&\sqrt{T}\\\sqrt{T}&i\sqrt{R}\end{smallmatrix}\right)=\mathbb{R}_{\{s,t\}}\,.
\end{eqnarray}
Case $\tau_s<\tau_t=\tau$. Then we have $\tau_s'\,=\tau_t'=\nicefrac{\tau^{\text{($st$)}}}{2}=\nicefrac{\tau}{2}=\tau'$, and consequently
\begin{eqnarray}
\mathbb{L}_{\{s,t\}}=\left(\begin{smallmatrix}i\sqrt{R}&\sqrt{T}\\\sqrt{T}&i\sqrt{R}\end{smallmatrix}\right)\!\left(\begin{smallmatrix}0\,&0\vspace{0.1cm}\\0\,&1\end{smallmatrix}\right)=\mathbb{R}_{\{s,t\}}\,.
\end{eqnarray}
(For the case $\tau_t<\tau_s=\tau$ similar reasoning holds).

Having verified identity in Eq.\,(\ref{Delta-commutation}) we can write
\begin{eqnarray}\nonumber
\!\!\!\!\!\!\!\!\!\!\Delta_{\tau'}\,\vec{u}'
&\stackrel{(\ref{u'})(\ref{Delta-commutation})}{=}&\Bigg(\prod_{j\in\mathcal{D}}(\mathbb{1}-\mathbb{P}_j)\prod_{k\in\mathcal{S}}\mathbb{S}_k\!\prod_{\{s,t\}\in\mathcal{B}}\!\mathbb{B}_{st}\Bigg)\ \Delta_{\tau}\,\vec{u}\\
&\stackrel{(\ref{abc})}{\sim}&\prod_{j\in\mathcal{D}}(\mathbb{1}-\mathbb{P}_j)\prod_{k\in\mathcal{S}}\mathbb{S}_k\!\prod_{\{s,t\}\in\mathcal{B}}\!\mathbb{B}_{st}\ \vec{z}\,,
\end{eqnarray}
which proves that after transformation the following condition holds
\begin{eqnarray}\label{c}
c)\ \ \ \Delta_{\tau'}\,\vec{u}'\sim\vec{z}\,'\,,
\end{eqnarray}
with vector $\vec{z}\,'$ given in Eq.\,(\ref{lemma-z'}).

Finally, let us check position of the particle after the transformation if we know that at the beginning it is in path $q=i$. Clearly, change of the path is possible only on a beam splitter and hence we have two cases.

In the case when $i$-th path does not go into a beam splitter (i.e. for $i\notin\bigcup\mathcal{B}$) the particle remains in the path $q\ \longrightarrow\ q'=q=i$. Hence, along with Eq.\,(\ref{T_prime}) we get
\begin{eqnarray}\label{ab}
a)\ \ q'=i\,,\ \ \ &
b)\ \ \tau_i'=\tau'>0\,.
\end{eqnarray}
A quick comparison of Eqs.~(\ref{c}) and (\ref{ab}) with definitions in Eqs.~(\ref{Lambda_zi}) and (\ref{[z]i}) reveals that in that case distribution $\bm{p}\in[\,\vec{z}\,]_i$ gets transformed into a distribution with support in $\Lambda_{\vec{z}\,'}^i$ meaning that $\bm{p}'\in[\,\vec{z}\,']_i$. It proves Eq.\,(\ref{zz}) with vector $\vec{z}\,'$ given in Eq.\,(\ref{lemma-z'}).

Otherwise, if $i$-th path crosses with another path on the beam splitter $\mathbb{B}_{st}$, i.e. for $i\in\{s,t\}\in\mathcal{B}$, the particle my change its position either to path $s$ or $t$ with the respective probabilities specified by Eq.\,(\ref{Gates-B-q}), i.e. we get
\begin{eqnarray}\label{q'st}
q\longrightarrow \Bigg\{\begin{array}{ll}q'=s&\,\text{with probability\ \ }\tfrac{|u_s'|^2}{|u_s'|^2+|u_t'|^2}\,,\vspace{0.1cm}\\q'=t&\,\text{with probability\ \ }\ \tfrac{|u_t'|^2}{|u_s'|^2+|u_t'|^2}\,.\end{array}
\end{eqnarray}
Let us note that in this case Eqs.~(\ref{B_prime}) and (\ref{T_prime}) entail that
\begin{eqnarray}\label{b'}
b)\ \ \ \ \tau_s'=\tau_t'=\tau'>0\,.
\end{eqnarray}
In particular, it means that $\Delta_{\tau'}$ restricted to entries $\{s,t\}$ equals identity. Therefore, from Eq.\,(\ref{c}) we have
\begin{eqnarray}
\begin{pmatrix}u_s'\\u_t'\end{pmatrix}\sim\begin{pmatrix}z_s'\\z_t'\end{pmatrix}\,,
\end{eqnarray}
which allows to replace $u$'s with $z$'s in Eqs.~(\ref{q'st}) whenever $|z_s'|^2+|z_t'|^2\neq0$. We thus obtain
\begin{eqnarray}\label{a'}
a)\ \ \ \ \Bigg\{\begin{array}{ll}q'=s&\,\text{with probability\ \ }\ \tfrac{|z_s'|^2}{|z_s'|^2+|z_t'|^2}\,,\vspace{0.1cm}\\q'=t&\,\text{with probability\ \ }\ \tfrac{|z_t'|^2}{|z_s'|^2+|z_t'|^2}\,.\end{array}
\end{eqnarray}
This result along with Eqs.~(\ref{c}) and (\ref{b'}) should be compared with definitions in Eqs.~(\ref{Lambda_zi}) and (\ref{[z]i}). We conclude that in the case of beam splitter placed in $i$-th path the system initially described by distribution $\bm{p}\in[\,\vec{z}\,]_i$ after the transformation has support in $\Lambda_{\vec{z}\,'}^s\cup\Lambda_{\vec{z}\,'}^t$ with cumulative probability over the respective sets given by Eq.\,(\ref{a'}). In other words, the system is described by a probabilistic mixture
\begin{eqnarray}
\bm{p}'=\tfrac{|z_s'|^2}{|z_s'|^2+|z_t'|^2}\,\bm{p}_s'+\tfrac{|z_t'|^2}{|z_s'|^2+|z_t'|^2}\,\bm{p}_t'\,,
\end{eqnarray}
with $\bm{p}_s'\in[\,\vec{z}\,']_s$ and $\bm{p}_t'\in[\,\vec{z}\,']_t$ which proves Eq.\,(\ref{zz-BS}) with vector $\vec{z}\,'$ given in Eq.\,(\ref{lemma-z'}).

\end{proof}

\section{Proof of Theorem~\ref{Theorem-Matrix}}

Now, we are ready to prove our main result.

\begin{proof}[\textbf{Proof of Theorem~\ref{Theorem-Matrix}}]\ \vspace{0.1cm}

Let us take $\bm{p}\in[\,\vec{z}\,]$. Recall that Definition~\ref{Definition} Eq.\,(\ref{[z]}) specifies distributions  in $[\,\vec{z}\,]$ as convex combinations of distributions in auxiliary classes $\bm{p}_i\in[\,\vec{z}\,]_i$, i.e. we have
\begin{eqnarray}\label{p}
\bm{p}=\sum_{i=1}^Np_i\,\bm{p}_i\,,
\end{eqnarray}
with
\begin{eqnarray}\label{pp}
p_i=|z_i|^2&\ \ \text{and}\ \ &\text{supp}\,\bm{p}_i\subset\Lambda_{\vec{z}}^i\,.
\end{eqnarray}
Clearly, since $\vec{z}$ is normalised we have $\sum_ip_i=\Vert\vec{z}\Vert^2=1$.

In the following we are interested in the shape of distribution $\bm{p}'$ which obtains from $\bm{p}$  as result of transformation via parallel configuration of gates $\mathcal{F}$, $\mathcal{D}$, $\mathcal{S}$, $\mathcal{B}$. For the proof of Theorem~\ref{Theorem-Matrix} we will find convex decomposition of $\bm{p}'$ into a mixture of distributions $\bm{p}_i'\in[\,\vec{z}\,'\,]_i$ and then compare it with Definition~\ref{Definition} Eq.\,(\ref{[z]}).

We begin with two simple observations about distribution $\bm{p}$ in Eq.\,(\ref{p}). Since distributions $\bm{p}_i$ have disjoint supports, see Eq.\,(\ref{ij-disjoint}), and $q=i$ only for the ontic state $(q,\vec{u},\vec{\tau})$ in support of $\bm{p}_i$\,, then we have
\begin{eqnarray}\label{prob-q=i}
&&\textit{\textsf{Pr}}\,(q=i)=p_i\,,\\\label{prob|q=i}
&&\textit{\textsf{Pr}}\,(q,\vec{u},\vec{\tau}|q=i)=\bm{p}_i(q,\vec{u},\vec{\tau})\,.
\end{eqnarray}
Furthermore, the sum in Eq.\,(\ref{p}) can be split into four groups by collecting together the terms associated with the same kind of gate attached to the relevant path, i.e.
\begin{eqnarray}\label{p-split}
\bm{p}=\sum_{l\in\mathcal{F}}p_l\,\bm{p}_l+\sum_{j\in\mathcal{D}}p_j\,\bm{p}_j+\sum_{k\in\mathcal{S}}p_k\,\bm{p}_k+\sum_{r\in\bigcup\mathcal{B}}p_r\,\bm{p}_r\,,\ \ \ \ 
\end{eqnarray}
since $\mathcal{F}$, $\mathcal{D}$, $\mathcal{S}$, $\bigcup\mathcal{B}$ partition the set of paths $\{1,\dots,N\}$.

Throughout the proof we will use the following auxiliary notation
\begin{eqnarray}\label{w}
\vec{w}':=\prod_{j\in\mathcal{D}}(\mathbb{1}-\mathbb{P}_j)\prod_{k\in\mathcal{S}}\mathbb{S}_k\!\prod_{\{s,t\}\in\mathcal{B}}\!\mathbb{B}_{st}\ \,\vec{z}\,,
\end{eqnarray}
which relates to vector $\vec{z}\,'$ in Eqs.~(\ref{Theorem-Matrix-Negative}) and (\ref{lemma-z'}) as follows
\begin{eqnarray}\label{wz}
\vec{z}\,'=\nicefrac{\vec{w}'}{\Vert\vec{w}'\Vert}\,.
\end{eqnarray}
Recall that in both Theorem~\ref{Theorem-Matrix} and Lemma~\ref{Lemma} vector $\vec{z}\,'$ is normalised.$\,^{\ref{Footnote-Norm}}$

\vspace{0.2cm} 
$\bullet$ {\it First part -- Eqs.~(\ref{Theorem-Matrix-Prob}) and (\ref{Theorem-Matrix-Dj}).}\ \vspace{0.1cm}

Form the fact that detectors react only to particles in a given path we get that detector $D_j$ '\textsc{Clicks}' if and only if $q=j$, which form Eqs.~(\ref{pp}) and (\ref{prob-q=i}) happens with probability
\begin{eqnarray}
\textit{\textsf{\,Pr}}\,(D_j|\vec{z})=\textit{\textsf{\,Pr}}\,(q=j)=|z_j|^2\,.
\end{eqnarray}
Note that because the system is always in a well-defined ontic state (which means that position of the particle is fixed) simultaneous detection in different detectors at the same time (i.e. more than one '\textsc{Click}') is impossible. Moreover, negative outcome in all detectors '\textsc{NoClick}' occurs only if $q\notin\mathcal{D}$, which happens with probability $1-\sum_{j\in\mathcal{D}}\textit{\textsf{\,Pr}}\,(q=j)=1-\sum_{j\in\mathcal{D}}|z_j|^2$.

To conclude, we observe that '\textsc{Click}' in detector $D_j$ provides additional knowledge that $q=j$. It entails update of the initial probability distribution $\bm{p}$ to the conditional distribution $\bm{p}_j\in[\,\vec{z}\,]_j$ given in Eq.\,(\ref{prob|q=i}). Then from Lemma~\ref{Lemma} Eq.\,(\ref{Lemma-Di}) we infer that after detection (post-selection on a '\textsc{Click}' in detector $D_j$) the system is described by a distribution $\bm{p}'\in[\,\bm{e}_j\,]_j$. Since in this case from Definition~\ref{Definition} Eq.\,(\ref{[z]}) we have $[\,\bm{e}_j\,]_j=[\,\bm{e}_j\,]$, then
\begin{eqnarray}
[\,\vec{z}\,]\ni\bm{p}\ \xymatrix{\ar[r]^{D_j} &}\ \bm{p}'\in[\,\bm{e}_j\,]&\ \ &('\textsc{Click}')\,
\end{eqnarray}
It proves the first part of Theorem~\ref{Theorem-Matrix} Eqs.~(\ref{Theorem-Matrix-Prob}) and (\ref{Theorem-Matrix-Dj}).

\vspace{0.2cm} 
$\bullet$ {\it Second part -- Eq.\,(\ref{Theorem-Matrix-Negative}).}\ \vspace{0.1cm}

Now we consider the case when none of the detectors '\textsc{Clicks}' (i.e. $q\notin\mathcal{D}$) or there is no detectors at all (i.e. $\mathcal{D}=\varnothing$). Observe that since distributions $\bm{p}_i\in[\,\vec{z}\,]_i$ in Eq.\,(\ref{p-split}) are supported in $\Lambda_{\vec{z}}^i$ for which $q=i$, then additional knowledge of $q\notin\mathcal{D}$ entails update of the initial probability distribution $\bm{p}$ to the following form
\begin{eqnarray}\label{p-split-2}
\bm{p}\ \longrightarrow\ \tilde{\bm{p}}&=&\sum_{l\in\mathcal{F}}\tilde{p}_l\,\bm{p}_l+\sum_{k\in\mathcal{S}}\tilde{p}_k\,\bm{p}_k+\sum_{r\in\bigcup\mathcal{B}}\tilde{p}_r\,\bm{p}_r\,,\ \ \ \ \ \ \ \ 
\end{eqnarray}
with renormalised coefficients
\begin{eqnarray}\label{p-tilde}
\tilde{p}_i=\tfrac{p_i}{\sum_{j\notin\mathcal{D}}p_j}\stackrel{(\ref{pp})}{=}\tfrac{|z_i|^2}{\sum_{j\notin\mathcal{D}}|z_j|^2}\,,
\end{eqnarray}
where $i\in\mathcal{F}\cup\mathcal{S}\cup\bigcup\mathcal{B}$. For further convenience let us rewrite the last sum more explicitly
\begin{eqnarray}\label{p-split-tilde}
\tilde{\bm{p}}&=&\sum_{l\in\mathcal{F}}\tilde{p}_l\,\bm{p}_l+\sum_{k\in\mathcal{S}}\tilde{p}_k\,\bm{p}_k+\!\!\!\!\sum_{\{s,t\}\in\mathcal{B}}\big(\tilde{p}_s\,\bm{p}_s+\tilde{p}_t\,\bm{p}_t\big)\,.\ \ \ \ \ \ \ \
\end{eqnarray}

In the following we are interested in the transformation of $\tilde{\bm{p}}$ under action of the parallel configuration of gates. Since supports of distributions $\bm{p}_i$ are disjoint, we can individually analyse each term in Eq.\,(\ref{p-split-tilde}) and then collect together all the results.

From Lemma~\ref{Lemma} Eq.\,(\ref{zz}) for $l\in\mathcal{F}$ and $k\in\mathcal{S}$ we get
\begin{eqnarray}\label{proof-F-S}
\begin{aligned}
\tilde{p}_l\,\bm{p}_l&\ \longrightarrow\ \tilde{p}_l\,\bm{p}_l'\,,\\
\tilde{p}_k\,\bm{p}_k&\ \longrightarrow\ \tilde{p}_k\,\bm{p}_k'\,,
\end{aligned}
\end{eqnarray}
where $\bm{p}_l'\in[\,\vec{z}\,']_l$ and $\bm{p}_k'\in[\,\vec{z}\,']_k$ with label $\vec{z}\,'$ given in Eq.\,(\ref{lemma-z'}). As for the coefficients $\tilde{p}_i$ defined in Eq.\,(\ref{p-tilde}), we observe that
\begin{eqnarray}\label{||z||}
\begin{aligned}
\!\!\sum_{j\notin\mathcal{D}}|z_j|^2&=\Big\Vert\prod_{j\in\mathcal{D}}(1-\mathbb{P}_j)\,\vec{z}\,\Big\Vert^2\\
&=\Big\Vert\prod_{k\in\mathcal{S}}\mathbb{S}_k\!\!\prod_{\{s,t\}\in\mathcal{B}}\!\!\mathbb{B}_{st}\,\prod_{j\in\mathcal{D}}(\mathbb{1}-\mathbb{P}_j)\,\vec{z}\,\Big\Vert^2\stackrel{(\ref{w})}{=}\Vert\vec{w}'\Vert^2\,,
\end{aligned}
\end{eqnarray}
where the first equality comes from definition of projectors $\mathbb{P}_j$, the second one is due to preservation of norm under $\prod_{k\in\mathcal{S}}\mathbb{S}_k\prod_{\{s,t\}\in\mathcal{B}}\mathbb{B}_{st}$ which is a unitary transformation, and the third draws on commutativity of all factors in the product. Now, from the fact that the product in Eq.\,(\ref{w}) consists of factors acting separately in the respective components, we have
\begin{eqnarray}\label{z-z}
\begin{aligned}
&\ \ \,z_l\,=\,w_l'\ \ \ \ \ \ \,\text{for}\ \ l\in\mathcal{F}\,,\\
&|z_k|\,=\,|w_k'|\ \ \ \ \text{for}\ \ k\in\mathcal{S}\,.
\end{aligned}
\end{eqnarray}
Substitution of Eqs.~(\ref{||z||}) and (\ref{z-z}) into Eq.\,(\ref{p-tilde}) gives
\begin{eqnarray}\label{proof-first-p}
\begin{aligned}
&\tilde{p}_l=\tfrac{|w_l'|^2}{\Vert\vec{w}'\Vert^2}\stackrel{(\ref{wz})}{=}|z_l'|^2=: p_l'\,,
\\
&\tilde{p}_k=\tfrac{|w_k'|^2}{\Vert\vec{w}'\Vert^2}\stackrel{(\ref{wz})}{=}|z_k'|^2=: p_k'\,.
\end{aligned}
\end{eqnarray}
Therefore, for the first two sums in Eq.\,(\ref{p-split-tilde}) we get
\begin{eqnarray}\label{proof-first}
\sum_{l\in\mathcal{F}}\tilde{p}_l\,\bm{p}_l+\sum_{k\in\mathcal{S}}\tilde{p}_k\,\bm{p}_k\ \longrightarrow\ \sum_{l\in\mathcal{F}}p_l'\,\bm{p}_l'+\sum_{k\in\mathcal{S}}p_k'\,\bm{p}_k'\,.\ \ \ \ \ 
\end{eqnarray}

Now, we proceed to the analysis of terms in the last sum in Eq.\,(\ref{p-split-tilde}). From Lemma~\ref{Lemma} Eq.\,(\ref{zz-BS}) we obtain (note that $\tilde{p}_s+\tilde{p}_t\neq0$ entails $|z_s'|^2+|z_t'|^2\neq0$)
\begin{eqnarray}\label{proof-BS}
\begin{aligned}
\tilde{p}_s\,\bm{p}_s+\tilde{p}_t\,\bm{p}_t\ \longrightarrow\ \ &\tilde{p}_s\Big(\tfrac{|z_s'|^2}{|z_s'|^2+|z_t'|^2}\,\bm{q}_s'+\tfrac{|z_t'|^2}{|z_s'|^2+|z_t'|^2}\,\bm{q}_t'\Big)+\\
&\tilde{p}_t\Big(\tfrac{|z_s'|^2}{|z_s'|^2+|z_t'|^2}\,\bm{q}_s''+\tfrac{|z_t'|^2}{|z_s'|^2+|z_t'|^2}\,\bm{q}_t''\Big)\ ,\ \ \ 
\end{aligned}
\end{eqnarray}
with $\bm{q}_s',\bm{q}_s''\in[\,\vec{z}\,']_s$ and $\bm{q}_t',\bm{q}_t''\in[\,\vec{z}\,']_t$ and the label $\vec{z}\,'$ given by formula Eq.\,(\ref{lemma-z'}). By regrouping terms on the right hand side we get
\begin{eqnarray}
\tfrac{|z_s'|^2}{|z_s'|^2+|z_t'|^2}\big(\tilde{p}_s\bm{q}_s'+\tilde{p}_t\bm{q}_s''\big)+\tfrac{|z_t'|^2}{|z_s'|^2+|z_t'|^2}\big(\tilde{p}_s\bm{q}_t'+\tilde{p}_t\bm{q}_t''\big)\,,\ \ \ \ \end{eqnarray}
and observe that 
\begin{eqnarray}
\bm{p}_s':=\tfrac{\tilde{p}_s\bm{q}_s'+\tilde{p}_t\bm{q}_s''}{\tilde{p}_s+\tilde{p}_t}&\ \ \ \text{and}\ \ \ &\bm{p}_t':=\tfrac{\tilde{p}_s\bm{q}_t'+\tilde{p}_t\bm{q}_t''}{\tilde{p}_s+\tilde{p}_t}
\end{eqnarray}
are properly normalised distributions with the property that $\bm{p}_s'\in[\,\vec{z}\,']_s$ and $\bm{p}_t'\in[\,\vec{z}\,']_t$. Therefore Eq.\,(\ref{proof-BS}) can be rewritten in the form
\begin{eqnarray}\label{proof-BS-z}
\tilde{p}_s\,\bm{p}_s+\tilde{p}_t\,\bm{p}_t&\longrightarrow&\tfrac{|z_s'|^2(\tilde{p}_s+\tilde{p}_t)}{|z_s'|^2+|z_t'|^2}\,\bm{p}_s'+\tfrac{|z_t'|^2(\tilde{p}_s+\tilde{p}_t)}{|z_s'|^2+|z_t'|^2}\,\bm{p}_t'\,.\ \ \ \ \ \ \ \ \ 
\end{eqnarray}
A closer look at the first coefficient reveals that 
\begin{eqnarray}\label{proof-last-p}
\begin{aligned}
&\!\!\!\!\!\!\!\!\!\!\tfrac{|z_s'|^2(\tilde{p}_s+\tilde{p}_t)}{|z_s'|^2+|z_t'|^2}\stackrel{(\ref{p-tilde})}{=}\tfrac{|z_s'|^2}{\sum_{j\notin\mathcal{D}}|z_j|^2}\cdot\tfrac{|z_s|^2+|z_t|^2}{|z_s'|^2+|z_t'|^2}\\
&\!\!\!\!\!\!\!\!\!\!\stackrel{(
\ref{wz})(\ref{||z||})}{=}\tfrac{|w_s'|^2}{\Vert\vec{w}'\Vert^2}\cdot\tfrac{|z_s|^2+|z_t|^2}{|w_s'|^2+|w_t'|^2}=\tfrac{|w_s'|^2}{\Vert\vec{w}'\Vert^2}\stackrel{(
\ref{wz})}{=}|z_s'|^2=:p_s'\,,
\end{aligned}
\end{eqnarray}
where in the penultimate equality the last fraction cancels out. The latter is due to the fact that  the only nontrivial action on components $\{s,t\}$ in Eq.\,(\ref{w}) comes from the matrix $\mathbb{B}_{st}$ which gives
\begin{eqnarray}
\begin{pmatrix}w_s'\\w_t'\end{pmatrix}=\begin{pmatrix}i\sqrt{R}&\sqrt{T}\\\sqrt{T}&i\sqrt{R}\end{pmatrix}\begin{pmatrix}z_s\\z_t\end{pmatrix}\,,
\end{eqnarray}
and since it is a unitary transform it preserves the norm $|w_s'|^2+|w_t'|^2=|z_s|^2+|z_t|^2$. Clearly, the same reasoning applies to the second term in Eq.\,(\ref{proof-BS-z}) which equals to $|z_t'|^2=:p_t'$. Hence for the last sum in Eq.\,(\ref{p-split-tilde}) we get
\begin{eqnarray}\label{proof-last}
\sum_{\{s,t\}\in\mathcal{B}}\big(\tilde{p}_s\,\bm{p}_s+\tilde{p}_t\,\bm{p}_t\big)\ \longrightarrow\sum_{\{s,t\}\in\mathcal{B}}\big(p_s'\,\bm{p}_s'+p_t'\,\bm{p}_t'\big)\,.\ \ \ \ \ 
\end{eqnarray}

Having analysed all terms in Eq.\,(\ref{p-split-tilde}) we use results in Eqs.~(\ref{proof-first-p}), (\ref{proof-first}), (\ref{proof-last-p}) and (\ref{proof-last}) which together with Eq.\,(\ref{p-split-2}) provide the following decomposition
\begin{eqnarray}
{\bm{p}}\ \longrightarrow\ \bm{p}'&=&\sum_{i=1}^Np_i'\,\bm{p}_i'\,,
\end{eqnarray}
where $\bm{p}_i'\in[\,\vec{z}\,']_i$ and $p_i'=|z_i'|^2$ with vector $\vec{z}\,'$ given  in Eq.\,(\ref{lemma-z'}). (Clearly, for $i\in\mathcal{D}$ we have $p_i'=|z_i'|^2=0$.)
Comparing with Definition~\ref{Definition} Eq.\,(\ref{[z]}), we conclude that
\begin{eqnarray}
[\,\vec{z}\,]\ni\bm{p}\ \xymatrix{\ar[r] &}\ \bm{p}'\in[\,\vec{z}\,']\,,
\end{eqnarray}
which proves the second part of Theorem~\ref{Theorem-Matrix} Eq.\,(\ref{Theorem-Matrix-Negative}).

\end{proof}

\end{document}